\PassOptionsToPackage{dvipsnames,table}{xcolor}
\PassOptionsToPackage{hypertexnames=false}{hyperref}
\documentclass[acmsmall,screen,natbib,nonacm]{acmart}

\setcopyright{acmlicensed}
\copyrightyear{2025}
\acmYear{2025}
\acmDOI{}

\author{Feras A.~Saad}
\orcid{0000-0002-0505-795X}
\affiliation{%
  \institution{Carnegie Mellon University}
  \city{Pittsburgh}
  \country{USA}
}

\authorsaddresses{}

\usepackage{adjustbox}
\usepackage{afterpage}
\usepackage{amsmath}
\usepackage{amsthm}
\usepackage{mathtools}
\usepackage{stmaryrd}
\usepackage{bm}
\usepackage{booktabs}
\usepackage{comment}
\usepackage{graphicx}
\usepackage{multicol}
\usepackage{multirow}
\usepackage{nicematrix}
\usepackage{soul}
\usepackage{siunitx}
\usepackage[most]{tcolorbox}
\usepackage[mathscr]{eucal}
\usepackage{accents} 
\usepackage[textsize=tiny]{todonotes}

\usepackage{tikz}
\usepackage{tikz-qtree}
\usetikzlibrary{arrows}
\usetikzlibrary{arrows.meta}
\usetikzlibrary{backgrounds}
\usetikzlibrary{calc}
\usetikzlibrary{decorations.pathreplacing}
\usetikzlibrary{intersections}
\usetikzlibrary{fit}
\usetikzlibrary{patterns}
\usetikzlibrary{positioning}
\usetikzlibrary{shapes}
\usetikzlibrary{shadows}
\usetikzlibrary{shapes.geometric}
\usetikzlibrary{tikzmark}

\usepackage[inline]{enumitem}
\setlist[itemize,1]{wide=\parindent,leftmargin=*}

\usepackage{zref-clever}
\NewDocumentCommand{\cref}{m}{\zcref{#1}}
\NewDocumentCommand{\Cref}{m}{\zcref[cap=true]{#1}}
\NewDocumentCommand{\crefrange}{mm}{\zcref[range,rangetopair=false]{#1,#2}}
\NewDocumentCommand{\Crefrange}{mm}{\zcref[cap=true,range,rangetopair=false]{#1,#2}}
\zcsetup{nameinlink=true}
\zcLanguageSetup{english}{cap,rangesep=\textendash}

\usepackage{algorithm}
\usepackage[indLines=true,noEnd=true,commentColor=darkgray,endLComment={}]{algpseudocodex}
\algnewcommand{\IfThenElse}[3]{\algorithmicif\ #1\ \algorithmicthen\ #2\ \algorithmicelse\ #3}
\makeatletter\AddToHook{env/algorithmic/begin}{\zcsetup{currentcounter=ALG@line}}\makeatother
\zcsetup{countertype={ALG@line=line}}
\zcRefTypeSetup{line}{nocap}
\zcRefTypeSetup{equation}{nocap}
\zcRefTypeSetup{equation}{Name-sg=Eq.,name-sg=,Name-pl=Eqs.,name-pl=,Name-sg-ab=,name-sg-ab=,Name-pl-ab=,name-pl-ab=}
\zcRefTypeSetup{item}{Name-sg=,name-sg=item,Name-pl=,name-pl=,Name-sg-ab=,name-sg-ab=,Name-pl-ab=,name-pl-ab=}
\newif\ifzcsectioncap
\newcommand{\zcsectionmarkcapTrue}{\global\zcsectioncaptrue}
\newcommand{\zcsectionmarkcapFalse}{\global\zcsectioncapfalse}
\newcommand{\zcsectionprefix}{\ifzcsectioncap\else\S\fi}
\newcommand{\zcsectionrangeprefix}{\ifzcsectioncap\else\S\S\fi}
\zcRefTypeSetup{section}{
    nocap,
    Name-sg = {\zcsectionmarkcapTrue Section~},
    Name-pl = {\zcsectionmarkcapTrue Sections~},
    name-sg = {\zcsectionmarkcapFalse},
    name-pl = {\zcsectionmarkcapFalse},
    namesep = {},
    refbounds = {,\zcsectionprefix,,},
    +refbounds-rb = {,\zcsectionrangeprefix,,},
    +refbounds-re = {,,,},
    pairsep = { and~},
    listsep = {, },
    lastsep = {, and~},
    rangesep = {\textendash},
}

\usepackage{keytheorems}
\newkeytheorem{theorem}[style=acmplain,numberwithin=section]
\newkeytheorem{lemma}[style=acmplain,sibling=theorem]
\newkeytheorem{proposition}[style=acmplain,sibling=theorem]
\newkeytheorem{corollary}[style=acmplain,sibling=theorem]
\newkeytheorem{definition}[style=acmdefinition,sibling=theorem,qed=\textup{\guillemotleft}]
\newkeytheorem{remark}[style=acmdefinition,sibling=theorem,qed=\textup{\guillemotleft}]
\newkeytheorem{notation}[style=acmdefinition,sibling=theorem,qed=\textup{\guillemotleft}]
\newkeytheorem{observation}[style=acmdefinition,sibling=theorem,qed=\textup{\guillemotleft}]
\newkeytheorem{example}[style=acmdefinition,sibling=theorem,qed=\textup{\guillemotleft}]

\usepackage{newfloat}
\usepackage[subrefformat=parens]{subcaption}
\DeclareFloatingEnvironment[name=Listing]{listing}
\zcsetup{
  countertype    = { sublisting = listing },
  counterresetby = { sublisting = listing },
}

\usepackage{doi}
\usepackage{url}
\usepackage[most]{tcolorbox}
\makeatletter
\tcbset{
  nicebox-base/.style={
    enhanced,
    sharp corners,
    colback=white,
    colframe=black,
    boxrule=0.5pt,
    colbacktitle=pink,
    coltitle=black,
    top=-1mm,
    bottom=-1mm,
    boxsep=2.0mm,
    toptitle=0.0mm,
    after app={\@endparenv},
  }
}
\newtcolorbox{nicebox}[1][]{
  nicebox-base,
  #1,
  after app={\@endparenv},
}
\makeatother

\makeatletter
\newtcolorbox{quotebox}[1][]{
    sharp corners,
    size=tight,
    colback=white,
    colframe=black,
    boxrule=0pt,
    leftrule=3pt,
    left=2.5mm,
    left skip=2\parindent,
    #1,
    after app={\@endparenv},
}

\usepackage{listings}
\lstdefinelanguage{verb}{
  basicstyle=\ttfamily\small,
  morecomment=[l]{\#},
  commentstyle=\color{darkgray}\ttfamily,
  morekeywords=[1]{if, else, sample, observe, score, return},
  keywordstyle=[1]\bfseries,
  numberstyle=\tiny,
  captionpos=b,
  escapeinside={<@}{@>},
}

\makeatletter
\NewTCBListing{codeverb}{!O{} !O{}}{
  sharp corners,
  size=tight,
  colback=white,
  colframe=black,    
  boxrule=0pt,       
  leftrule=3pt,      
  left skip=2\parindent,    
  top=-2mm,          
  bottom=-2mm,       
  left=2.5mm,        
  listing only,
  listing options={language=verb,#2},
  #1,
  after app={\@endparenv},
}
\makeatother

\makeatletter
\patchcmd{\NAT@test}{\else \NAT@nm }{\else \NAT@nmfmt{\NAT@nm}}{}{}
\makeatother
\setcitestyle{citesep={,},notesep={;\,}}

\allowdisplaybreaks

\NewDocumentCommand{\R}{o}{\IfNoValueTF{#1}{\mathbb{R}}{\mathbb{R}^{#1}}}
\NewDocumentCommand{\N}{}{\mathbb{N}}

\NewDocumentCommand{\Leb}{O{}}{\mathsf{L}\IfBlankF{#1}{^{#1}}}
\NewDocumentCommand{\Counting}{O{}}{\#\IfBlankF{#1}{_{#1}}}
\NewDocumentCommand{\Haus}{m o o}{\IfNoValueTF{#2}%
  {\mathsf{H}^{#1}\IfNoValueF{#3}{\!\restriction_{#3}}}%
  {\mathsf{H}_{#2}^{#1}\IfNoValueF{#3}{\!\restriction_{#3}}}}

\RenewDocumentCommand{\d}{}{\mathop{}\!\mathrm{d}}

\NewDocumentCommand{\Borel}{O{}O{}}{\mathbf{B}\IfBlankF{#1}{\!\left(#1\right)}\IfBlankF{#2}{\!\restriction #2}}
\NewDocumentCommand{\Powset}{m}{2^{#1}}

\NewDocumentCommand{\restr}{m m}{#1\!\restriction_{#2}}

\DeclareMathOperator{\rank}{rank}

\DeclarePairedDelimiter{\set}{\lbrace}{\rbrace}
\DeclarePairedDelimiter{\abs}{\lvert}{\rvert}
\DeclarePairedDelimiter{\setr}{\lparen}{\rparen}
\DeclarePairedDelimiter{\setb}{\lbrack}{\rbrack}
\DeclarePairedDelimiter{\setl}{\langle}{\rangle}

\newcommand{\defeq}{\coloneq}

\newcommand{\given}{\,|\,}

\newcommand{\dx}{\d{x}}
\newcommand{\dz}{\d{z}}

\newcommand{\cO}{\Omega}
\newcommand{\sigO}{\mathcal{F}}
\newcommand{\baseO}{\mathbb{P}}
\newcommand{\sigE}{\mathcal{E}}
\newcommand{\sigG}{\mathcal{G}}

\NewDocumentCommand{\Prob}{om}{\mathbb{P}\IfValueT{#1}{_{#1}}\IfBlankF{#2}{\setb*{#2}}}

\NewDocumentCommand{\Expect}{sO{}m}{\IfBooleanTF{#1}{\widehat{\mathbb{E}}}{\mathbb{E}}\IfValueT{#2}{_{#2}}\setb*{#3}}

\renewcommand{\Pr}[1]{\mathbb{P}\setb*{#1}}

\NewDocumentCommand{\ind}{sO{}m}{\IfBooleanTF{#1}
  {\mathbf{1}\IfBlankF{#2}{_{#2}}\IfBlankF{#3}{(#3)}}
  {\mathbf{1}\IfBlankF{#2}{_{#2}}\IfBlankF{#3}{\!\left(#3\right)}}
  }

\NewDocumentCommand{\dirac}{d<>mO{}}{\delta\IfValueT{#1}{^#1}_{#2}\IfBlankF{#3}{\setr*{#3}}}

\newcommand{\pushf}[2]{#1_{\#}#2}

\NewDocumentCommand{\X}{o}{\IfNoValueTF{#1}{\mathsf{X}}{\mathsf{X}_{#1}}}

\NewDocumentCommand{\SigX}{o}{\IfNoValueTF{#1}{\Sigma_{\X}}{\Sigma_{\X}^{#1}}}

\NewDocumentCommand{\lamX}{o}{\IfNoValueTF{#1}{\lambda_{\X}}{\lambda_{\X}^{#1}}}

\NewDocumentCommand{\muX}{o}{\IfNoValueTF{#1}{\mu_{\X}}{\mu_{\X}^{#1}}}

\NewDocumentCommand{\pX}{so}{\IfBooleanTF{#1}
  {\IfNoValueTF{#2}{\tilde{p}_{\X}}{\tilde{p}_{\X}^{#2}}}
  {\IfNoValueTF{#2}{p_{\X}}{p_{\X}^{#2}}}
  }

\NewDocumentCommand{\Z}{o}{\IfNoValueTF{#1}{\mathsf{Z}}{\mathsf{Z}_{#1}}}

\NewDocumentCommand{\Eobs}{O{\pi}}{E_{#1}}

\NewDocumentCommand{\SigZ}{o}{\IfNoValueTF{#1}{\Sigma_{\Z}}{\Sigma_{\Z}^{#1}}}

\NewDocumentCommand{\lamZ}{o}{\IfNoValueTF{#1}{\lambda_{\Z}}{\lambda_{\Z}^{#1}}}

\NewDocumentCommand{\muZ}{o}{\IfNoValueTF{#1}{\mu_{\Z}}{\mu_{\Z}^{#1}}}

\NewDocumentCommand{\pZ}{o}{\IfNoValueTF{#1}{p_{\Z}}{p_{\Z}^{#1}}}

\NewDocumentCommand{\lamXZ}{so}{\IfBooleanTF{#1}
  {\IfNoValueTF{#2}{\eta_{\X\mid\Z}}{\eta_{\X\mid\Z}^{#2}}}
  {\IfNoValueTF{#2}{\lambda_{\X\mid\Z}}{\lambda_{\X\mid\Z}^{#2}}}
  }

\NewDocumentCommand{\muXZ}{o}{\IfNoValueTF{#1}{\mu_{\X\mid\Z}}{\mu_{\X\mid\Z}^{#1}}}

\NewDocumentCommand{\muZX}{o}{\IfNoValueTF{#1}{\mu_{\Z\mid\X}}{\mu_{\Z\mid\X}^{#1}}}

\NewDocumentCommand{\pXZ}{o}{\IfNoValueTF{#1}{p_{\X\mid\Z}}{p_{\X\mid\Z}^{#1}}}

\NewDocumentCommand{\pZX}{o}{\IfNoValueTF{#1}{p_{\Z\mid\X}}{p_{\Z\mid\X}^{#1}}}

\NewDocumentCommand{\Apath}{O{\pi}}{A_{#1}}
\NewDocumentCommand{\kpath}{O{\pi}}{k_{#1}}
\NewDocumentCommand{\gpath}{O{\pi}}{g_{#1}}
\NewDocumentCommand{\dpath}{O{\pi}}{d_{#1}}
\NewDocumentCommand{\npath}{O{\pi}}{n_{#1}}
\NewDocumentCommand{\mpath}{O{\pi}}{m_{#1}}

\NewDocumentCommand{\U}{O{\pi}}{U_{#1}}

\NewDocumentCommand{\B}{O{\pi}}{B_{#1}}

\NewDocumentCommand{\obs}{so}{\IfBooleanTF{#1}
  {\IfNoValueTF{#2}{\bar{t}}{\bar{t}_{#2}}}
  {\IfNoValueTF{#2}{t}{t_{#2}}}
  }

\NewDocumentCommand{\obsI}{so}{\IfBooleanTF{#1}
  {\IfNoValueTF{#2}{\bar{t^{-1}}}{\bar{t^{-1}}_{#2}}}
  {\IfNoValueTF{#2}{t^{-1}}{t^{-1}_{#2}}}
  }

\NewDocumentCommand{\Jac}{O{\pi,a}}{J_{#1}}

\newcommand{\sample}{\textbf{sample}}
\newcommand{\observe}{\textbf{observe}}
\newcommand{\score}{\textbf{score}}
\newcommand{\interval}{\textbf{interval}}
\newcommand{\return}{\textbf{return}}

\newcommand{\distr}[1]{\textsf{#1}}
\newcommand{\cPDF}{\textsf{pdf}}

\NewDocumentCommand{\DistRepo}{}{\mathsf{Dist}}
\newcommand{\kDD}{\mathbf{d}}        
\newcommand{\kDC}{\mathbf{c}}        

\NewDocumentCommand{\Idisc}{O{\pi}}{I_{#1}^{\kDD}}
\NewDocumentCommand{\Icont}{O{\pi}}{I_{#1}^{\kDC}}

\NewDocumentCommand{\kindD}{O{D}}{\kappa_{#1}}
\NewDocumentCommand{\ThetaD}{O{D}}{\Theta_{#1}}
\NewDocumentCommand{\CarrierD}{O{D}}{M_{#1}}
\NewDocumentCommand{\SigD}{O{D}}{\Sigma_{#1}}
\NewDocumentCommand{\lamD}{O{D}}{\lambda_{#1}}
\NewDocumentCommand{\muD}{O{D}}{\mu_{#1}}
\NewDocumentCommand{\pD}{O{D}}{p_{#1}}
\NewDocumentCommand{\rdimD}{O{D}}{r_{#1}}

\NewDocumentCommand{\Zle}{m m}{\Z^{(#1,\le #2)}}
\NewDocumentCommand{\Zeq}{m m}{\Z^{(#1,#2)}}

\NewDocumentCommand{\rd}{O{\pi}}{r_{#1}}

\NewDocumentCommand{\QXZ}{o}{Q_{\X\mid z}\IfNoValueF{#1}{^{#1}}}

\begin{document}

\title{Bayes Rule for Conditioning Probabilistic Programs}

\begin{abstract}
Probabilistic programs provide a rigorous programming-language foundation
for expressing rich probabilistic models and offer language-level support
for solving challenging statistical inference problems.
A key theoretical challenge that probabilistic programs inherit from
probability theory is conditioning a probability distribution
on an observation that has probability zero.
Many languages address this challenge by using computational rules for
conditioning that may work well for some problems but are statistically
unsound for others.

This article introduces a principled mathematical approach for
conditioning probabilistic programs.
We first show that existing techniques based on ``soft conditioning'' and
``neighborhood limits'' exhibit semantic issues from the perspective
of measure theory and produce inference results that disagree
with sound statistical reasoning.
By leveraging tools from geometric measure theory, we construct a novel
disintegration-based semantics for probabilistic program traces that
enables sound conditioning.
The theory delivers an explicit Bayes rule representation of the
conditional distribution, which is directly compatible with widely used
probabilistic inference algorithms that require access to unnormalized
posterior densities.
Our Bayes rule handles probabilistic models with loops, branches,
discrete, continuous, and mixed-type distributions, observations on
manifolds, differentiable structure, and additional modeling
motifs that arise in practice.
\end{abstract}

\maketitle


\section{Introduction}
\label{sec:introduction}

Probabilistic programs are standard computer programs augmented with a
construct to \textit{sample} random variables from probability distributions.
These languages can be used to model and analyze complex probabilistic
processes, and have been thoroughly studied by the programming languages
community~\cite{barthe2020}.
Within this field is an important class of languages for
\textit{statistical probabilistic programming}, which enjoy broad adoption
among the machine learning, statistics, and data science communities~\citep{jwdvm2021}.
Statistical probabilistic programs include an additional construct to
\textit{observe} random variables taking on particular values,
which conditions the probability distribution
induced by the sample statements to agree with a set of constraints.
These systems provide sophisticated probabilistic inference engines to
solve statistical problems, such as generating samples from the conditional distribution
or estimating conditional expectations.

The goal of this paper is to build a rigorous mathematical
foundation for conditioning probabilistic programs that aligns with
principled statistical theory.
To illustrate the fundamental challenge, consider a hidden random variable
$X$ and an observable random variable $Z$.
For discrete random variables, the conditional distribution
of $X \mid Z=z$ is given by the usual Bayes rule:
\begin{equation}
\Prob{X = x \mid Z = z} = \frac{\Prob{X=x, Z=z}}{\Prob{Z=z}},
\label{eq:intro-bayes-rule}
\end{equation}
which is well-defined whenever $\Prob{Z=z} > 0$.
For discrete probability, this theory is complete:
any semantics for discrete probabilistic programs can safely
ignore probability-zero observations, because conditioning on
$z$ for which $\Prob{Z=z}=0$ is statistically undefined
without controversy.

Real-world applications, however, demand the use of continuous random
variables such as the normal, uniform, or gamma, which take on
uncountably many outcomes.
In this case, \textit{every} observation $z$
satisfies $\Prob{Z=z}=0$, which rules out the use of the standard Bayes rule.
This challenge with conditioning on probability-zero events has a long
history in probability theory.

Modern probabilistic programming languages use variants of the Bayes
rules to condition on probability-zero events, which can be problematic.
The first issue is semantic: while the computational behavior of the
program is fully precise, the inference problem being solved may have no
reasonable statistical interpretation.
The second issue is practical: because inference algorithms such as
importance sampling, variational inference, and Markov chain Monte Carlo
all make use of unnormalized target distributions, using ad-hoc Bayes rules
can give unsound statistical results.

\subsection{Conditioning Probabilistic Programs using Probability Densities}
\label{sec:introduction-observe}
Students of elementary probability theory learn to handle continuous
random variables by replacing probabilities such as $\Prob{X=x}$
with probability densities such as $\pX(x)$.
The result is a modified Bayes rule that is used
for models with both continuous and discrete random variables:
\begin{equation}
\pXZ(x \given z) = \frac{\pX(x) \pZX(z \given x)}{\pZ(z)}.
\label{eq:intro-bayes-rule-pdf}
\end{equation}
This formula combines probability mass functions for discrete variables
and probability density functions for continuous variables,
which is the standard approach adopted by
probabilistic programming languages whose semantics treat observations
as ``soft conditions''~\citep{staton2016semantics,borgstrom2016lambda}.
Many probabilistic programming users are surprised to find that soft
conditioning cannot handle simple observation patterns, such as in
the program shown in \cref{fig:exact-condition}.

\begingroup
\setlength{\normalparindent}{\parindent}
\setlength{\intextsep}{6pt}
\begin{figure}[H]
\begin{minipage}[m]{.5\linewidth}
\begin{codeverb}[left skip=\normalparindent]
x = sample normal(0,1)
observe x*x == 4
\end{codeverb}
\end{minipage}\hfill
\begin{minipage}[m]{.5\linewidth}
\includegraphics[width=\linewidth]{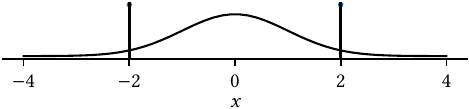}
\end{minipage}
\captionsetup{skip=0pt}
\caption{Probabilistic program conditioning on a function of a random variable.}
\label{fig:exact-condition}
\end{figure}
\endgroup

In this program, we have $X \sim \distr{Normal}(0,1)$ and
$Z = \obs(X) = X^2 \ge 0$.
Statistical theory tells us that, by symmetry of the normal distribution,
the conditional distribution of $X \mid Z=z$ must be an equally
weighted mixture of atoms at $\obsI(z) \defeq \set{-\sqrt{z}, +\sqrt{z}}$.
Also, as $Z$ is fixed given $X$, it is conditionally a discrete random variable
with mass $\pZX(z \given x) = \ind{x^2 = z}$.
With this setup, applying the modified Bayes rule fails for
\textit{every} value of $z$, because the denominator is zero:
\begin{equation*}
\pZ(z)
= \int_{\X} \pZX(z\given x) \pX(x) \d x
= \int_{\X} \ind{\obs(x) = z} \pX(x) \d x
= \int_{\X} \ind{x \in \set{-\sqrt{z}, +\sqrt{z}}} \pX(x) \d x
= 0.
\end{equation*}

For an alternative approach, we recognize that $Z$ is marginally a continuous
$\distr{ChiSq}(1)$ random variable, so it
does have a well-defined probability density function $\pZ(z) = e^{-z/2}/\sqrt{2 \pi z} $.
Applying the mantra of "masses for discrete variables" and "densities for continuous variables",
we can use this continuous density $\pZ(z)$ in the denominator of the Bayes rule and
set the conditional mass $\pXZ(x \given z)$ in the numerator to be a discrete distribution on $\set{-\sqrt{z},+\sqrt{z}}$.
In this case, we have
\begin{align}
\pXZ(x \given z) = \frac{ e^{-x^2/2} }{e^{-z/2}/\sqrt{z}} \ind{x^2 = z}
\implies
\pXZ(-2 \given 4) = \pXZ(2 \given 4) = 2.
\end{align}

The result is clearly not a valid discrete distribution for $X$, as the
probabilities of the two atoms sum up to 4.
Standard probabilistic programming inference engines exhibit this failure
in different ways.
Some systems will silently fail, by sampling traces from the prior distribution
and assigning them a likelihood weight of zero.
Other systems force the user to wrap the observation around artificial
noise, e.g., \lstinline[language=verb]|x = normal(0,1); observe normal(x*x, 0.001) = 4|,
which fundamentally changes the original model desired by the user and in
turn alters the resulting statistical inferences.

The limitations of soft conditioning are not restricted to the situation
just demonstrated above (known as \textit{exact conditioning}~\citep{stein2024exact}).
Inferences using soft conditioning fail to comport with statistical theory in the presence
of mixed-type random variables---which are discrete in one branch and
continuous in another---or in the presence of observe statements within
conditionally executed program branches.
More generally, as we argue in \cref{sec:observe}, language semantics based
on probability densities are largely inadequate for developing a coherent
theory for statistical probabilistic programs.
They are best understood as languages for defining unnormalized measures,
which carry a meaningful statistical interpretation in only a restricted
fragment in the language.

\subsection{Conditioning Probabilistic Programs via Neighborhood Limits}
\label{sec:introduction-neighborhood}

To address the limitations of soft conditioning,
\citet{jacobs2021} introduces an alternative approach based on arguments
from \citet{jaynes2003}.
The key idea is to formalize conditioning on a probability zero set $E$ in terms of a
limiting process, such that for any event $A \subset \X$,
\begin{align}
\Prob{X \in A \mid E} \defeq
  \lim_{\epsilon \to 0}
  \frac
    {\Prob{X \in A \cap E_\epsilon}}
    {\Prob{E_\epsilon}},
\label{eq:intro-neighborhood}
\end{align}
where $E_\epsilon$ is a family of positive probability events that
descend to $E$.
An unavoidable consequence of the Borel--Kolmogorov paradox is that the limit can
and will depend on the chosen family of limiting events, so
\citeauthor{jacobs2021} provides a language construct for users to specify
the neighborhoods within the $\observe$ statements.
\Citeauthor{jacobs2021} demonstrates that this approach can resolve
certain ``paradoxes'' with the standard soft conditioning semantics, such
as unit errors, handling mixed-type distributions, and observe statements
in conditional branches.
The language also supports observations of transformations of random
variables (but these functions are assumed to be injective, and so the
observation pattern with $Z = \obs(X) = X^2$ from
\cref{fig:exact-condition} is not handled directly).

In \cref{sec:neighborhood}, we study the viability of neighborhood limits
as a semantic foundation for statistical probabilistic programs.
By analyzing the measure-theoretic structure of limiting arguments,
we find that this approach exhibits similar limitations to
soft conditioning using probability densities.

In essence, the language of \citeauthor{jacobs2021} can be understood as
enabling a user to condition a probabilistic program on a parametric family
of measure-zero events, say $E(z)$ for observations $z \in \Z$.
We show that null sets in the language do not necessarily arise from observing
a unified random variable $Z = \obs(X)$, i.e., it is not the case
that $E(z) = \set{Z=z}$ for every $z \in \Z$.
This high degree of flexibility puts the formalism of \citeauthor{jacobs2021}
in tension with the standard setting of statistical inference procedures,
which fundamentally revolve around a \textit{fixed} measurable function of the data:
\begin{quotebox}\label{quote:schervish}
A \textit{statistic} is virtually any measurable function of the data, $X$.

\hfill---\citet[p.~83]{schervish1995}.
\end{quotebox}
For the fragment of the language where the null sets
$E(z) = \set{Z=z}$ arise from a unified statistic $Z = \obs(X)$,
then conditioning by limits can be given a rigorous statistical
interpretation (\cref{sec:mlp-neighborhood}).
In the general case, however, neighborhood limits only provide a computational
rule for defining specific conditional probabilities $\Prob{A \given E}$ at fixed
events $A,E$: these individual numbers do not unify
into a coherent conditional probability distribution in the standard statistical sense
(cf.~\cref{prop:limit-not-probability}).

\subsection{Conditioning Probabilistic Programs via Disintegration}
\label{sec:introduction-disintegration}

The approach to conditioning advocated for this article is based
on \textit{disintegration}, a probabilistic concept
that generalizes the elementary law of total probability.
In the discrete setting, a \textit{conditional distribution} of $X \given Z$ is any
function $z \mapsto \lambda x.~\Pr{X=x \given Z=z}$ that factorizes
the prior distribution:
\begin{align}
\Pr{X = x} = \sum_{z \in \Z} \Pr{X = x \mid Z= z} \Pr{Z=z},
&&
\mbox{where }
\Pr{X \in \obsI(z) \mid Z = z} = 1.
\label{eq:intro-disintegrate}
\end{align}
This factorization gives an \textit{inversion} of the generative model.
To generate a trace $X$ from the law $\Pr{X=\cdot}$, we may first sample $Z$
from its marginal distribution $\Pr{Z=\cdot}$ and then sample an $X$
such that $Z = t(X)$ from the
conditional distribution $\Pr{X = \cdot \given Z=z}$.
The marginal distribution of $Z$ is
\begin{align}
\Pr{Z=z} = \sum_{x \in \X} \ind{t(x)=z}\Pr{X=x}.
\end{align}
For each $z$ with $\Prob{Z=z}>0$, any conditional distribution
$\lambda x.~\Pr{X = x \mid Z= z}$
agrees with the usual Bayes rule~\cref{eq:intro-bayes-rule}.
Otherwise if $\Prob{Z=z}=0$, the identity~\cref{eq:intro-disintegrate} shows
that the conditional distribution is undetermined, exactly as in
the usual Bayes rule for discrete probabilities.

Disintegration generalizes this theory from discrete probability spaces to general
spaces.
For example, a disintegration $X \mid Z=z$
of the probabilistic program from \cref{fig:exact-condition}
is given by
\begin{codeverb}
b = sample bernoulli(0.5)
if b == 0
  sample dirac(sqrt(z))
else
  sample dirac(-sqrt(z))
\end{codeverb}
In general, a disintegration is a family of probability measures
$\Prob{X \in \cdot \given Z=z}$ for $z \in \Z$ such that
\begin{align}
\Prob{X \in A} = \int_{\Z} \Prob{X \in A \given Z=z} \Prob{Z=\d z},
&&
\mbox{where }
\Pr{X \in \obsI(z) \mid Z = z} = 1.
\label{eq:intro-disintegration-general}
\end{align}

This identity shows that $\Prob{X \in \cdot \given Z=z}$
is undetermined for every $z$ such that $\Prob{Z=z} = 0$, precisely as in the discrete case.
Instead what matters is the \textit{expected} behavior of the program, over
random observations $Z$ drawn from the data generating process.
The statistical issues with the soft-conditioning approach from
\cref{sec:introduction-observe} and the neighborhood-limit approach
from \cref{sec:introduction-neighborhood} arise from
defining the entire semantics of a probabilistic program in terms
of quantities that are arbitrary on a measure-zero set.
Indeed, it is impossible to assign good statistical meaning when
conditioning a generative model on a \textit{specific} null set,
as famously stated by \citeauthor{kolmogorov1933}:
\begin{quotebox}
The concept of a conditional probability with regard to an isolated hypothesis
whose probability equals 0 is inadmissible.
\hfill---\normalfont{\citet[p.~51]{kolmogorov1933}}
\end{quotebox}

This fact is reinforced by decades of attempts by various researchers to
resolve the Borel--Kolmogorov paradox: if $(X_1,X_2)$ are jointly
continuous in $\R[2]$, there is simply no canonical way to condition on the null event
$\set{X_1=X_2}$ (see \citet{bungert2022lion} for a recent formidable attempt).
On the other hand, for statistical inference procedures with statistic
$Z = \obs(X)$, the disintegration formula furnishes a fully coherent semantics for the
conditional law $X \mid Z=z$.

\subsection{This Work: A Measure-Theoretic Foundation for Conditioning by Disintegration}
\label{sec:intro-contributions}

Equipped with this background, we establish a principled
measure-theoretic foundation for conditioning statistical probabilistic
programs by disintegration.
For a flexible class of probabilistic programs (essentially Turing
complete, equipped with exact real computation), we seek to identify
conditions on the prior distribution over traces
$X \in \X$ and the observation map $\obs: \X \to \Z$ that guarantee the
existence of a disintegration $\muXZ$.
Using techniques from geometric measure theory~\citep{federer1996},
we develop a disintegration for a large class of probabilistic
programs that include loops, branches, discrete, continuous, and mixed-type
random variables, manifold-supported and mixed-dimensional variables, differentiable
structure, and other constructs used in practice.

Beyond existence of the disintegration, we identify an explicit
representation in terms of a Bayes rule that addresses the
limitations from using the formulas for soft conditioning~\cref{eq:intro-bayes-rule-pdf} and
neighborhood limits~\cref{eq:intro-neighborhood}.
The Bayes rule form of the disintegration can be used for a variety of approximate
probabilistic inference algorithms---such as importance sampling, Markov
chain Monte Carlo, variational inference, and sequential Monte Carlo---that
require access to the target density only up to an (unknown) normalizing
constant.
As a proof-of-concept, we prove the consistency of likelihood weighting in a
probabilistic language with unbounded while loops and trace addresses.

\paragraph{Scope}
The scope of this article is firmly on the mathematical side.
We do not seek to build a language implementation that reifies the key
ideas of our construction, and focus primarily on measure-theoretic rather
than language-design aspects.
Our objective is to establish a rigorous foundation for
statistical probabilistic programs that
\begin{enumerate*}[label=(\roman*)]
\item gives insight into existing disintegration-based systems that are
formalized using operational rules; and
\item enables the development of new probabilistic languages that
include more flexible constructs supported by our theory.
\end{enumerate*}
While the theoretical development behind the soundness arguments is
relatively involved, our Bayes rule for conditional distributions is
readily interpretable: it can be used by language designers and
practitioners as an operational device, without reference to its
measure-theoretic justification.

\paragraph{Contributions}

In summary, the main contributions of this article are as follows.

\begin{enumerate}[label={\textbf{(C\arabic*)}}, wide=0pt, itemsep=4pt, leftmargin=*]

\item We systematically review the ways in
  which ``soft conditioning'' semantics using densities
  lead to ill-behaved programs and produce incorrect statistical inference results
  (\cref{sec:observe}).

\item We analyze the measure-theoretic
  structure behind conditioning via neighborhood limits. While
  this procedure provides a workable computational rule, we prove that it does not
  guarantee a coherent posterior distribution over the probability space
  and does not comport with the semantic soundness properties of
  conditioning used in the statistics literature (\cref{sec:neighborhood}).

\item We construct a
  measure-theoretic foundation for inference over probabilistic program traces.
  We prove a novel disintegration theorem for a large family of
  statistical models and show that it carries an explicit Bayes
  rule that can be used by inference algorithms (\crefrange{sec:disintegration}{sec:mlp}).
\end{enumerate}

Language design considerations are discussed in \cref{sec:design} and
related work in \cref{sec:related}.


\section{Conditioning Probabilistic Programs using Probability Densities}
\label{sec:observe}

The dominant paradigm in probabilistic programming---adopted
by major languages such as Gen~\citep{towner2019}, Pyro~\citep{bingham2019},
Stan~\citep{carpenter2017}, and PyMC~\citep{abril2023pymc}---uses soft conditioning~\citep{borgstrom2016lambda,staton2016semantics,huot2023}
to express statistical inference problems.
Our goal in this section is to highlight how
soft conditioning with probability density functions produces ill-defined
programs and inference results that disagree with statistical theory.
These case studies serve to motivate the alternative theory of
conditioning that we develop in \crefrange{sec:lang}{sec:mlp} to resolve
the limitations of soft conditioning.

\begin{listing}[H]
\caption{Language for generative probabilistic programs with observe statements.}
\label{lst:prob-while}
\begin{nicebox}[size=minimal,before skip=5pt, after skip=5pt, boxrule=1pt,bottom=0pt,top=0pt]
\setlength{\abovedisplayskip}{0pt}
\begin{align*}
x &\in \mathsf{Vars}\quad
r \in \mathsf{Real}\quad
o \in \mathsf{Op}\quad
D \in \mathsf{Dist}
\\
e &\in \mathsf{Expr} \Coloneq x \mid r \mid D \mid [e^*] \mid o(e) \mid D(e) \mid \sample~e \\
s &\in \mathsf{Stmt} \Coloneq
  x \boldsymbol{=} e
  \mid \observe~e_1~e_2
  \mid \mathbf{skip}
  \mid s_1~\textbf{;}~s_2
  \mid \textbf{if}~e~\textbf{then}~s_1~\textbf{else}~s_2
  \mid \textbf{while}~e~\textbf{do}~s
\end{align*}
\end{nicebox}
\end{listing}

\Cref{lst:prob-while} shows a standard probabilistic language with soft conditioning,
which includes a $\sample$ expression and an $\observe$ statement.
A $\sample~D(\theta)$ expression generates a random variable from
a probability distribution $D(\theta)$ whose parameters are $\theta$;
for example $D(\theta) = \distr{Normal}(0,1)$ for a continuous variable
and $D(\theta) = \distr{Poisson}(4)$ for a discrete variable.
Similarly, $\observe~D(\theta)~z$ statement records
a ``measurement'' of a random draw from $D(\theta)$ to take value $z$.
Typing rules for expressions are entirely standard
(e.g., in $\textbf{sample}~e$, the expression $e$
must evaluate to a distribution name $D \in \mathrm{Dist}$
applied to valid parameters, etc.) and are omitted.
To highlight the essential statistical issues, which is the focus of this paper,
we restrict attention to programs that halt \textit{almost surely}.
Such programs still contain an infinite number of control-flow paths,
e.g., using a \textbf{while} loop to sample a geometric random variable, or
sampling a Poisson number $N$ of loop iterations.

\subsection{Semantics of Straight-Line Programs}
\label{sec:observe-slp}

A straight-line probabilistic program denotes a target probability density
over sampled variables, weighted by the observations.
In particular, let $\pi$ be a straight-line program that dynamically encounters
$k$ $\sample$ statements over the course of its execution,
say $\sample~D_i(\theta_i(x_{1:i-1}))$ for $i=1,\dots,k$.
The parameters $\theta_i$ may depend
on all $x_{1:i-1}$ previously generated samples in the execution, as well
as other constants in the current environment.
Further, suppose there are $u$ $\observe$ statements,
say $\observe~D'_j(\theta_j(x_{1:n_j}))~\varphi_j(x_{1:n_j})$ for $j=1,\dots,u$,
where again the distributional parameters $\theta_j$ and observed values
$z_j \defeq \varphi_j(x_{1:n_j})$
depend on all $n_j \in \set{0,\dots,k}$ previous samples.

With this notation, the probabilistic program denotes a probability density function
\begin{align}
\pX[\pi](x_{1:k})
&\defeq \frac{\pX*[\pi](x_1, \dots, x_k)}{C^\pi}
=  \frac{
      \prod_{i=1}^{k} h^\pi_i(x_i; x_{1:i-1}) \thinspace
      \prod_{j=1}^{u} g_j^\pi(z_j; x_{1:n_j})}
    {
      \int_{\R^{k}}\pX*[\pi](x_{1:k})\lamX[\pi](\d x_{1:k})
    },
\label{eq:observe-slp-ratio}
\shortintertext{whenever $C^\pi \in (0,\infty)$, where the densities of the sample and observe statement are defined by}
h_i^\pi(x_i; x_{1:i-1}) &\defeq \cPDF_{D_i(\theta_i(x_{1:i-1}))}(x_i), \qquad
g_j^\pi(z_j; x_{1:n_j}) \defeq \cPDF_{D'_j(\theta_j(x_{1:n_j}))}(\varphi_j(x_{1:n_j})).
\label{eq:observe-pdf}
\end{align}
Therefore, \cref{eq:observe-slp-ratio} is taken as a ``posterior density''
over the set $\X \defeq \R^k$ of traces,
with respect to the product $\lamX[\pi] \defeq \otimes_{i=1}^{k}\lamD[D_i]$ of base
measures (here $\lamD = \Leb$ is the Lebesgue measure if $D$ is a continuous distribution
and $\lamD = \Counting$ is the counting measure if $D$ is a discrete distribution).

\subsection{Inference with Densities}
Exact sampling from the distribution~\cref{eq:observe-slp-ratio} is
generally intractable.
Many algorithms for approximate inference bypass this challenge by working with
the unnormalized density $\pX*[\pi]$, which is available by running
the probabilistic program on a fixed trace and aggregating probability densities.

A standard example is sampling via ``likelihood weighting'', i.e.,
importance sampling using the prior as the proposal distribution.
We run the program to produce random samples $x_{1:k}$
from the $k$ $\sample$ statements, while carrying a weight $w$ from evaluating
the probability densities from the $\observe$ statements at their realized
parameters and measurements.
Given $J>0$ importance samples, we approximate the expectation of a test
function $f: \X \to \R$ under the target density \cref{eq:observe-slp-ratio}
using ratio estimation
(which is typically biased but consistent as $J \to \infty$):
\begin{equation}
\Expect[\pX[\pi]]{f\!\left(x_{1:k}\right)}
\approx
\textstyle
  {\sum_{j=1}^J f\!\left(x^{(j)}_{1:k}\right) w^{(j)}}
  \Big/
  {\sum_{j=1}^J w^{(j)}}.
\end{equation}

The unnormalized density can also be used for other inference algorithms
available in probabilistic programming libraries, such as Markov chain Monte
Carlo (MCMC), sequential Monte Carlo, or variational inference, which do
not require knowledge of the intractable normalizing constant $C^\pi$.

\subsection{Soft Conditioning and Statistical Inference in Generative Models}

The \sample{} and \observe{} keywords are constructs that suggestively
define a distribution by conditioning a generative model.
We would like to think of the \sample{} statements as defining ``hidden''
variables and the \observe{} statements as conditioning on the values
of ``observable'' variables.
\Cref{eq:observe-slp-ratio} resembles the Bayes rule, where the numerator
is the joint probability density of the hidden and observed variables and
the denominator is the marginal likelihood of the observed variables.

Despite these apparent parallels, the target
density~\cref{eq:observe-slp-ratio} in general carries no intrinsic meaning
as conditioning in the statistical sense.
A program with $\sample$ and $\observe$ is best understood as defining an
unnormalized target density: the statistical interpretation fails in two
places.
\begin{itemize}
\item First, the soft conditioning semantics depend on values of quantities
at specific null sets where measure-theory tells us that their values are
allowed to be arbitrarily defined.
As a result, a specific program may fail to denote a meaningful conditional
distribution, even when the underlying statistical inference problem is
perfectly well defined.

\item Second, under a statistical interpretation of programs with soft conditioning,
inference results computed using \zcref[cap=true]{eq:observe-slp-ratio}
do not cohere with inference results from statistical theory.
\end{itemize}

\subsection{Failure to Denote a Conditional Distribution}
\label{sec:observe-densities}

\subsubsection{Infinite Marginal Likelihood}
\label{sec:observe-densities-ml}

Consider the following probabilistic program.
\begin{codeverb}
u = sample uniform(0, 1)
x = sample gamma(u, 1)
observe normal(0, x) 0
\end{codeverb}
The program encodes an inference problem for the posterior distribution of
a Gaussian random variable $Z \sim \distr{Normal}(0,X)$ at $z = 0$, where
the variance $X$ has a Gamma prior whose shape $U$ has a uniform hyperprior.
The issue with this program is that its normalizing constant is infinite:
\begin{align*}
C^\pi
= \int_{0}^{1}
  \int_{0}^{\infty}\left[ \cPDF_{\distr{Normal}(0,x)}(z=0) \cdot \cPDF_{\distr{Gamma}(u,1)}(x)\right] \d x \d u
\propto \int_0^{1} \int_{0}^{\infty} \frac{x^{u-3/2} e^{-x}}{\Gamma(u)} \d x \d u = \infty.
\label{eq:heterozetesis}
\end{align*}
The inner integral is infinite
for every $u \in (0,0.5]$, which has positive measure under the uniform prior for $U$.
The same pathology arises if $u$ is replaced with any constant in
$(0,0.5]$ (see \cref{fig:uniform}).
\begin{figure}[H]
\includegraphics[width=\linewidth]{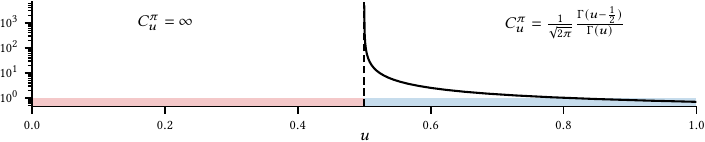}
\caption{Normalizing constant of the gamma-normal probabilistic program as a function of the variable $u$.}
\label{fig:uniform}
\end{figure}
The density semantics thus cannot produce
a well-defined distribution and fails to assign good meaning to the
program as a representation of the ``posterior distribution''
$(U, X) \given Z = 0$.

This pathology has a sound resolution under a statistical
interpretation of the program.
It turns out that the problematic term is not $u$, but
rather the program constant $0$ in the $\observe$ statement.
The density of a continuous distribution can be arbitrarily
defined on a set of measure zero, so the behavior of
$\cPDF_{\distr{Normal}(0,x)}(z)$ at $z=0$,
as well as the entire posterior distribution $(U,X) \given Z = 0$, are not
uniquely defined mathematical objects.
Second,
whatever representative of $\cPDF_{\distr{Normal}(0,x)}$ we choose,
any inference algorithm for the posterior $(U,X) \given Z$
has zero probability of encountering an observation $z$
in a ``bad'' set $A$ of points whose marginal likelihoods are invalid:
\begin{align}
\Prob{Z \in A} = 0, &&
A \defeq \set*{z \in \R \;\middle|\; C^\pi(z) \in \set{0,\infty} }.
\end{align}

But this resolution is incompatible with soft conditioning.
For example, the above program does not contain an explicit variable called
$Z$, which is represented only implicitly by an $\observe$ statement that
adds a likelihood factor to an unnormalized target density.
Any convention about a generative story for this $\observe$
statement is an external modeling choice that lives
outside the program's intrinsic semantics.
Moreover, the target density $\pX[\Pi](x_{1:k})$ over variables from
the $\sample$ statements is parameterized by the specific numerical
constants $z$ (in this case, $z=0$) written in the text of the program.
Such semantics are inadmissible from the statistical viewpoint.
The values of the probability density
functions $h^\pi_i$ and $g^\pi_j$ in \Cref{eq:observe-pdf}
are underdetermined at the program constants
that appear in $\observe$ statements.
As the semantics of the program depend on a specific $z$,
we can set $C_\pi$ to be any value by selecting another representative
of the density $\cPDF_{\distr{Normal}}$ at $z=0$.
There is no statistical basis on which to prioritize or refute any particular choice.

\subsubsection{Scoring with Arbitrary Values}
\label{sec:observe-densities-score}

Many languages feature a special construct called $\score(z)$,
which allows users to introduce an arbitrary multiplicative factor $z \in \R_{>0}$
to the unnormalized density.
While $\score$ appears more flexible than $\observe$ on the surface,
they are equivalent:
\begin{align}
\observe(D, z) \rightsquigarrow \score(\cPDF_D(z))
&&
\score(z) \rightsquigarrow \observe(\distr{Exponential}(z), 0).
\label{eq:equiv-observe-score}
\end{align}
The first expression holds by definition.
The second expression holds because
\begin{align}
\cPDF_{\distr{Exponential}(z)}(x=0) = ze^{-zx}\vert_{x=0} = z
&& (z > 0).
\end{align}

\Cref{eq:equiv-observe-score} gives us insight into the misbehavior
of additional nondenoting programs.
Consider the following program $\pi$, whose
normalizing constant
$C^\pi = \int \exp(x^2) \cPDF_{\distr{Normal}(0,1)}(x) \d x = \infty$.
\begin{codeverb}
x = sample normal(0,1)
score exp(x*x)
\end{codeverb}
This program does not include a numerical constant on which
probability densities are evaluated.
However, if we transform the $\score$ into an equivalent $\observe$ and interpret
the program as a generative model
$X \sim \distr{Normal}(0,1), E \sim \distr{Exponential}\left(\exp(X^2)\right)$,
we find that its denotation is equivalent to the posterior $X \mid E = e$ at $e=0$
(using the Bayes rule~\cref{eq:intro-bayes-rule-pdf} with canonical densities).
As in the previous case, the posterior is undetermined at $e=0$.
The statistical interpretation of this semantically equivalent program is again
tied to the behavior of a probability density at a null set.

More generally, there are many semantics-preserving transformations
from $\score$ to $\observe$ besides~\cref{eq:equiv-observe-score}.
For example, the fact for any mean $\mu \in \R$ and variance $v > 0$
\begin{equation}
\cPDF_{\distr{Normal}(\mu,v)}(x=\mu) = \frac{1}{\sqrt{2\pi v}}
\quad \mbox{ gives } \quad
\score(z) \rightsquigarrow \observe(\distr{Normal}\left(0,1/\left(2\pi z^2\right)\right), 0).
\end{equation}

In general, there could be many different generative stories behind
the $\observe$ statements in a given probabilistic program, which
all have the same density semantics.
It is therefore impossible to extract a statistical inference
problem from the text of the probabilistic program alone.

\subsubsection{Exact Conditions}
\label{sec:observe-densities-exact}

A statistical motif that $\sample$ and $\observe$ cannot express
is an ``exact'' condition.
For example, consider $X, Y \sim \distr{Uniform}(0,1)$ and $Z = 2X - Y$.
For every observation $z$ outside a null set,
the posterior distribution
$(X,Y)|\set{Z=z}$ is uniform over the line segment
$L(z) \defeq \set{(x,y) \in [0,1]^2 \mid 2x - y = z}$.
%
Many users attempt to express this problem as follows.
\begin{codeverb}
x = sample uniform(0,1)
y = sample uniform(0,1)
z = 2*x - y
observe bernoulli(z==0) 1
\end{codeverb}

With soft conditioning, the $\distr{Bernoulli}$ flip has
weight $0$ with probability $1$.
The weight of the observation is also $0$ with
probability $1$ and the normalizer is $C_\pi = 0$.
The $\sample$ and $\observe$ constructs are not flexible enough to handle
exact conditions, which are statistically well-defined but poorly defined
under the density semantics: the joint distribution of $(X,Y,Z)$ lives in
$\R^3$ but is singular with respect to $\Leb^3$, even though every pair of
variables is absolutely continuous with respect to $\Leb^2$.
The na\"ive choice of product base measures $\lamX[\pi]$ used to dominate
the density $\pX[\Pi]$ cannot reason about such
distributions, forcing $\pX*[\pi] = 0$ a.e.
The elementary Bayes rule~\cref{eq:intro-bayes-rule} assumes that the posterior distribution is
dominated by the prior, which is not the case with exact conditioning.

The issue here is dual to the issue from the previous two examples.
With exact conditions, the theoretical posterior distribution is
supported on a measure-zero set for almost every observation $z$.
However, the program encodes observation
as a Bernoulli variable---whose probability mass is \textit{never}
an arbitrary choice---that disagrees with its observed value with
prior probability one.

\subsection{Incorrect Statistical Inferences}
\label{sec:observe-inferences}

The previous section focused on semantic issues with probabilistic programs
when their denotations are sensitive to arbitrary choices at null sets.
We now focus on a well-behaved fragment where the target density does
denote a probability measure, but where
interpreting~\cref{eq:observe-slp-ratio} as a Bayes rule produces
incorrect statistical inference results.
These issues are best highlighted for probabilistic programs that exhibit
control flow (i.e., branching statements), whose semantics we briefly recount.

\subsubsection{Semantics of Programs with Control Flow}
\label{sec:observe-inferences-control}

For a probabilistic program with control flow,
enumerating all possible execution paths yields a
collection $\Pi$ whose elements are straight-line programs $\pi$.
Suppose that each straight-line program $\pi$ contains
$k_\pi$ $\sample$ statements and $u_\pi$ $\observe$ statements.
Let $A_\pi \subset \R^{k_\pi}$ be the set of sample traces
$x_{1:k_\pi}$ whose realized control flow is exactly equal to $\pi$, i.e., by
satisfying all the corresponding guard constraints in the $\textbf{if}$
and $\textbf{while}$ statements.
The trace space is $\X \defeq \bigcup_{\pi} \set{\pi} \times \R^{k_\pi}$
and the target probability density function is
\begin{equation}
\pX[\Pi](\pi, x_{1:k_\pi})
  \defeq \frac
    {\ind{x_{1:k_\pi} \in A_\pi} \pX*[^\pi](x_{1:k_\pi})}
    {\displaystyle \sum_{\gamma \in \Pi} \int_{\R^{k_\gamma}} \ind{x'_{1:k_\gamma} \in A_\gamma} \pX*[\gamma]\left(x'_{1:k_\gamma}\right)\lamX[\gamma]\left(\d x'_{1:k_\gamma}\right)},
\label{eq:observe-mlp-ratio}
\end{equation}
which is taken with respect to the coproduct
$\lamX[\Pi] \defeq \coprod_{\pi \in \Pi}\left(\dirac<\Pi>{\pi} \otimes \lamX[\pi]\right)$
of base measures.
The notation $\dirac<M>{m}$ denotes an
atomic measure on a set $M$ located at $m \in M$, often abbreviated
as $\dirac{m}$.

\subsubsection{Mixed-Type Distributions}
\label{sec:observe-inferences-mixed}

Consider a pair of random variables $(X,Z)$ where $X$ is a fair coin flip and
$Z \sim \distr{Uniform}(0,1)$ if $X = 0$ and $Z \sim \distr{Bernoulli}(0.5)$ otherwise.
Given an observation $\set{Z=1}$, the posterior $\Pr{X \mid Z = 1}$
contains no measure-theoretic difficulties as $\Pr{Z = 1} = 1/4 > 0$,
and so the posterior is not arbitrarily defined at this observation.
\Cref{lst:post-mixed-type} encodes this statistical inference problem
using a $\sample$ statement for $X$ and two $\observe$ statements
for $Z=1$ in different program branches.
The semantics~\cref{eq:observe-mlp-ratio} for probabilistic programs
with control flow gives
\begin{equation}
\pX[\Pi](x=1) = \frac{1/2 \cdot 1/2}{1/2 \cdot 1 + 1/2 \cdot 1/2} = \frac{1}{3},
\label{eq:post-mixture-ppl}
\end{equation}
i.e., the else-branch is \textit{less likely} than the if-branch.
However, statistical theory tells us that
the posterior distribution must assign probability one to the else-branch, because
a probability mass is infinitely more likely than a probability density.
An ad-hoc computation confirms this intuition: in order to compare the
likelihood of the observation under each branch as bona-fide probabilities
on equal footing, we take the limit of conditional probabilities that
$Z \in [1-\epsilon, 1+\epsilon]$
\begin{align}
\Pr{X=1 \mid Z = 1}
= \lim_{\epsilon \to 0} \Pr{X=1 \mid Z \in [1-\epsilon, 1+\epsilon]}
= \lim_{\epsilon \to 0} \left(\frac{1/2 \cdot 1/2}{1/2 \cdot 1/2 + 1/2 \cdot \epsilon}\right) = 1.
\label{eq:post-mixture-statistics}
\end{align}

\begin{listing}
\begin{sublisting}{.45\linewidth}
\begin{lstlisting}[language=verb, frame=single]
x = sample bernoulli(0.5)
if x == 0
  observe uniform(0,1) 1
else
  observe bernoulli(0.5) 1
\end{lstlisting}
\caption{Mixed-Type Distribution}
\label{lst:post-mixed-type}
\end{sublisting}\hfill
\begin{sublisting}{.45\linewidth}
\begin{lstlisting}[language=verb, frame=single]
x = sample bernoulli(0.5)
if x == 0
  observe uniform(0,1) 1
  observe uniform(0,1) 0.5
else
  observe uniform(0,1) 0.25
\end{lstlisting}
\caption{Control Flow with Different Observes}
\label{lst:post-control-flow}
\end{sublisting}
\caption{Programs where scoring semantics disagree with probabilistic conditioning.}
\label{lst:post}
\end{listing}

The disagreement between \cref{eq:post-mixture-ppl,eq:post-mixture-statistics}
is known in the literature, often as the ``Indian GPA Problem''~\citep{wu2018discrete}.
To clarify the issue, let $\pi_1,\pi_2 \in \Pi$ denote the two straight-line programs
in \cref{lst:post-mixed-type} corresponding to $X=0$ and $X=1$ respectively.
The probability densities for the $\observe$ statements
used in path densities $\pX[\pi_1]$ and $\pX[\pi_2]$
are defined with respect to \textit{different} base measures.
In particular, $Z$ is dominated by
$\Leb$ in path $\pi_1$ but by $\dirac{0} + \dirac{1}$
in path $\pi_2$.
\Cref{eq:observe-mlp-ratio} is therefore not a valid Bayes'
rule, which requires that the densities for hidden and observed
variables are defined with respect to a single unified base measure.
For the observable variable $Z$, a natural unified
base measure to use in the semantics
$\pX[\Pi]$ of the overall program
is $\Leb + \dirac{0} + \dirac{1}$, which gives
\begin{align}
\pi_1&: \Pr{Z \in \d z \mid X = 0} = \ind[(0,1)]{z} \cdot (\Leb(\d z) + \dirac{0}(\d z) + \dirac{1}(\d z)), \label{eq:z-dens-pi-1} \\
\pi_2&: \Pr{Z \in \d z \mid X = 1} = (1/2 \cdot \ind[\set{0,1}]{z}) \cdot (\Leb(\d z) + \dirac{0}(\d z) + \dirac{1}(\d z)). \label{eq:z-dens-pi-2}
\shortintertext{On the other hand, the implicit base measures used by the density semantics of $\observe$
are $\Leb$ for the path $\pi_1$ and $\dirac{0} + \dirac{1}$ for $\pi_2$:}
\pi_1&: \Pr{Z \in \d z \mid X = 0} = \ind[[0,1]]{z} \cdot \Leb(\d z) \label{eq:z-dens-pi-1-bad} \\
\pi_2&: \Pr{Z \in \d z \mid X = 1} = (1/2 \cdot \ind[\set{0,1}]{z}) \cdot (\dirac{0}(\d z) + \dirac{1}(\d z)). \label{eq:z-dens-pi-2-bad}
\end{align}

For path $\pi_1$, the correct base measure~\cref{eq:z-dens-pi-1}
forces the uniform density $\ind{z \in (0,1)}$ to \textit{exclude} $0$ and
$1$, while the incorrect base measure~\cref{eq:z-dens-pi-1-bad} allows the
density to be \textit{arbitrary} at $0$ and $1$.
From a language perspective, this example shows that the choice of correct base
measure is not local: it requires knowledge of the
distribution of the observable variables across all program paths in $\Pi$.
In \crefrange{sec:disintegration}{sec:mlp}, we develop a method to construct
correct base measures for general probabilistic programs that can be used to
soundly compare probability densities of observations across control flow paths.

\subsubsection{Continuous Observations Along Different Paths}

\Cref{lst:post-control-flow} shows a related issue to the mixed-type
distribution problem.
Under a statistical interpretation of this program,
there is a branching variable $X$.
If $X=0$ then we observe the uniform variables $(Z_1, Z_2) = (1, 0.5)$,
otherwise we observe the uniform variable $Z_1 = 0.25$.
Repeating the calculation from \cref{eq:post-mixture-ppl} gives the ``posterior'' probability
\begin{equation}
\pX[\Pi](x=1) = \frac{1/2 \cdot 1}{1/2 \cdot 1 + 1/2 \cdot 1 \cdot 1} = \frac{1}{2}.
\label{eq:post-branch-ppl}
\end{equation}
On the other hand, using a limit-based computation analogous to
\cref{eq:post-mixture-statistics} with neighborhoods
\begin{equation}
E_\epsilon \defeq \set*{X=0, Z_1 \in [1-\epsilon, 1+\epsilon], Z_2 \in [0.5-\epsilon, 0.5+\epsilon]} \cup \set*{X=1, Z_1 \in [0.25-\epsilon, 0.25+\epsilon] }
\label{eq:post-branch-neighborhood}
\end{equation}
shows that $\Pr{\set{X=0} \cap E_\epsilon}$ accumulates probability $\epsilon^2$
whereas $\Pr{\set{X=1} \cap E_\epsilon}$ accumulates $\epsilon$.
Only the
$X=1$ branch survives in the limit $\epsilon \to 0$ and has posterior probability one:
\begin{align}
\Pr{X = 1 \mid E}
= \lim_{\epsilon \to 0} \Pr{X=1 \mid E_\epsilon}
= \lim_{\epsilon \to 0} \frac{1/2  \cdot 2\epsilon}{1/2 \cdot \epsilon \cdot 2\epsilon + 1/2 \cdot 2 \epsilon}
= 1.
\end{align}
In general, under the statistical interpretation of the $\observe$ statements,
only those program paths $\pi \in \Pi$ with the fewest
continuous observations survive the limit using these neighborhoods.
Under the density semantics (where
$\observe$ is just a $\score$ with an arbitrary number), probabilities
are defined computationally using~\cref{eq:observe-mlp-ratio} and result in
unsound inferences such as \cref{eq:post-mixture-ppl,eq:post-branch-ppl}.

\medskip

\noindent The examples from
\cref{sec:observe-densities,sec:observe-inferences} emphasize
that the most coherent interpretation of the density semantics of a straight-line
probabilistic program $\pi$ is an unnormalized measure.
Its total mass $C_\pi$ is the expectation of the scores defined by the $\observe$ statements
under the probability distribution defined by the $\sample$ statements.
Control flow gives a sum of pathwise unnormalized measures, whose total
mass is the sum, over program paths, of the prior path probability times
the expected score conditional on taking that path.
Many users and language developers interpret
\cref{eq:observe-slp-ratio,eq:observe-mlp-ratio} as conditional probability
densities, which gives inference results that disagree with
statistical theory.


\section{Conditioning Probabilistic Programs Using Neighborhood Limits}
\label{sec:neighborhood}

To address the limitations of soft conditions---in particular those in
\cref{sec:observe-inferences} where density calculations produce incorrect
results---\citet{jacobs2021} uses a neighborhood-based
approach to conditioning on measure-zero events.
According to \citeauthor{jacobs2021}, the issue with the density
semantics is that they are not expressive enough to handle the inescapable
ambiguity that arises from conditioning on null sets.
To resolve these issues, \citeauthor{jacobs2021} introduces the construct
\begin{equation}
\observe(D(\theta), \interval(z, \epsilon)),
\end{equation}
which encodes an observation $Z \in [z-\epsilon, z+\epsilon]$
for some random variable $Z \sim D(\theta)$,
analogously to the approach we used
in \cref{eq:post-mixture-statistics,eq:post-branch-neighborhood}.
While the language does not include random variable names for the
$\observe$ statements, \citeauthor{jacobs2021} still interprets these statements in the
statistical modeling sense (i.e., not unnormalized measures),
as is evident in the discussion from \citet[\S3.2]{jacobs2021}.

Our goal in this section is to study neighborhood-limit semantics as a
candidate foundation for statistical probabilistic programming.
Indeed, \cref{sec:observe-inferences} showed that a carefully
chosen limit construction can result in plausible posterior inferences for
the programs in \cref{lst:post}.
We will show that this approach---when used in its full generality---exhibits
analogous issues to soft conditioning.

\subsection{Theory of Conditional Probability via Neighborhood Limits}
\label{sec:neighborhood-rao}

The following principle, adapted from \citet{jaynes2003}, forms the
basis of \citeauthor{jacobs2021}'s semantics.
\begin{quotebox}
\itshape
A general method for conditioning on measure-zero events is to define a
sequence of events $E_\epsilon$ parameterized by a number $\epsilon > 0$
such that $E_\epsilon$ in some sense converges to $E$ in the limit
$\epsilon \to 0$, but $\Prob{E_\epsilon} > 0$ for all $\epsilon > 0$.
We then \textbf{define} the conditional distribution to be the limit of
$\Prob{A \mid E_\epsilon}$:
\begin{align}
\Prob{A \mid E} = \lim_{\epsilon \to 0}\frac{\Prob{A \cap E_\epsilon}}{\Prob{E_\epsilon}}
\label{eq:conditional-probability-via-limits}
\end{align}
\hfill---\normalfont{\citet[p.~58:2]{jacobs2021}}
\end{quotebox}
The value of this limit will necessarily depend on the sequence of events
$E_\epsilon$, a consequence of the Borel--Kolmogorov paradox.
Any statistician who is truly concerned with conditional probabilities
given a bare null event $E$ \textit{must} also specify the neighborhoods
$E_\epsilon$ that converge to $E$.

\subsubsection{Limits Do Not Give Conditional Probability Distributions}

In \textit{Conditional Measures and Applications}, \citet[p.~68]{rao2005}
argues that the limit \cref{eq:conditional-probability-via-limits}
can be used to obtain a conditional distribution on the probability
space $(\cO, \sigO, \baseO)$ of a probabilistic experiment.
By appealing to the Vitali-Hahn-Saks theorem~\citep[Theorem 4.3.11]{renyi1970},
\citeauthor{rao1988} argues that if $E_1 \supset E_2 \supset \cdots$ is a sequence
of decreasing events such that $\cap_{n=1}^{\infty}E_n = E \ne \varnothing$,
with $\Prob{E}=0$ and
$\Prob{E_n} > 0$ for each $n$, then \cref{eq:conditional-probability-via-limits} defines
a probability distribution if the limit exists for all $A \in \sigO$.
If such a result holds, then the limit definition serves as a
rigorous approach to conditioning in probabilistic programs.

The flaw in \citeauthor{rao2005}'s reasoning (also in \citet[\S4b]{rao1988}) is that the hypothesis needed to invoke
Vitali-Hahn-Saks is impossible to satisfy using \cref{eq:conditional-probability-via-limits}:
the limit does not exist for all $A \in \sigO$.
The next result proves that the limit definition of
$\Prob{\cdot \mid E}$ does not give a conditional distribution.

\begin{proposition}
\label{prop:limit-not-probability}
Let $(\cO, \sigO, \baseO)$ be a probability space and let
$\set{E_n} \downarrow E$ be a sequence of events converging to $E$,
with $\Prob{E_n} > 0$ for every $n$ and $\Prob{E}=0$.
There is no measure $\nu$ on $\sigO$ such that
\begin{align}
\nu(A) = \lim_{n \to \infty} \frac{\Prob{A \cap E_n}}{\Prob{E_n}} && (A \in \sigO).
\end{align}
\end{proposition}

\begin{proof}
Assume that $\nu$ exists.
Clearly $\nu$ is finite, as $\nu(\Omega) = \lim_n \Prob{E_n}/\Prob{E_n}=1$.
Also
\begin{align}
\nu(E)
 \defeq \lim_{n \to \infty} \frac{\Prob{E \cap E_n}}{\Prob{E_n}}
  = \frac{\Prob{E}}{\Prob{E_n}} = 0,
\end{align}
because $E \subset E_n$ for every $n \ge 1$.
But computing $\nu(E)$ using continuity from above gives
\begin{align}
\nu(E)
= \nu(\cap_{k=1}^{\infty}E_k)
= \lim_{k \to \infty} \nu(E_k)
= \lim_{k \to \infty} \left( \lim_{n \to \infty} \frac{\Prob{E_k \cap E_n}}{\Prob{E_n}}  \right)
= \lim_{k \to \infty} \lim_{n \ge k} 1 = 1,
\end{align}
where the penultimate equality follows from the fact that
$E_n \subset E_k$ whenever $n \ge k$.
%
\end{proof}

\Cref{prop:limit-not-probability} implies, after specifying a limiting
neighborhood, the rule~\cref{eq:conditional-probability-via-limits}
does not deliver a coherent conditional distribution.
Rather, it delivers a collection of isolated numbers $\Prob{A \mid E}$ that
in general have no relationship with one another as the event $A \in \sigO$ changes.
%

\subsubsection{Conditional Probability Given a Sigma-Algebra}

As \cref{prop:limit-not-probability} rules out the possibility that a
neighborhood limit defines a conditional distribution, we will
analyze the mathematical object~\cref{eq:conditional-probability-via-limits}
from the lens of modern probability theory~\citep[\S5]{kolmogorov1933}.
Let $(\cO, \sigO, \baseO)$ be a probability space.
Given events $A, B \in \sigO$ with $\Prob{B}>0$, the conditional probability
$\Prob{A \mid B} \defeq \Prob{A \cap B}/\Prob{B}$ is defined in the usual way.
Given event $A \in \sigO$ and a sub-sigma algebra $\sigE \subset \sigO$,
we define the \textit{conditional probability}
$\Prob{A \given \sigE}$ to be any $\sigE$-measurable random variable such that
\begin{align}
\Prob{A \cap E} = \int_{E} \Prob{A \given \sigE}(\omega) \Prob{\d \omega}
&& (E \in \sigE).
\label{eq:conditional-probability}
\end{align}

As usual, the random variable $\Prob{A \given \sigE}$ is unique up to
$\baseO$-null sets.
For example, if $E \in \sigE$ is a
null event, then $\Prob{A \given \sigE}(\omega)$ can be defined arbitrarily
for $\omega \in E$, because $\Prob{A \cap E} = 0$ which leaves the integrand
unconstrained.
There is no mathematical basis on which to favor any particular version
over another, and program semantics must not
be sensitive to the chosen version.

Multiple individual conditional probabilities can be coherently combined
together using the idea of a ``regular'' conditional probability.
For a sub-sigma algebra $\sigG \subset \sigO$,
a \textit{regular conditional probability}
on $\sigG$ given $\sigE$ is a function
$\Prob{ \cdot \given \sigE }(\cdot): \sigG \times \cO \to [0,1]$
such that
\begin{itemize}[nosep]
\item For every $A \in \sigG$,
  the map $\omega \mapsto \Prob{ A \given \sigE }(\omega)$ is a conditional probability
  of $A$ given $\sigE$;
\item For $\baseO$-almost every $\omega \in \cO$, the map
  $A \mapsto \Prob{A \given \sigE}(\omega)$ is a probability measure on $(\cO, \sigG)$.
\end{itemize}

\subsubsection{The Sigma-Algebra Structure of a Neighborhood}

Let us now analyze neighborhood limits using conditional probabilities.
Following \citet[p.~68]{rao2005}, we identify
$\sigE \defeq \sigma(\set{E_n : n \ge 1})$
as the underlying sigma-algebra that
is being conditioned on when the practitioner specifies
a neighborhood-limit \cref{eq:conditional-probability-via-limits}
with descending sets $E_1 \supset E_2 \supset \cdots$.
This sigma-algebra $\sigE$ is generated by the following disjoint events,
which are all atoms of $\sigE$:
\begin{align}
B_0 \defeq \cO \setminus E_1,
\qquad\qquad
B_n \defeq E_n \setminus E_{n+1} \quad (n \ge 1),
\qquad\qquad
B_\infty \defeq E.
\end{align}
It is readily verified that any conditional distribution (regular or otherwise)
given $\sigE$ is of the form
\begin{align}
\Prob{A \given \sigE}(\omega) =
\sum_{n : \Prob{B_n} > 0} \Prob{A \given B_n} \ind[B_n]{\omega}
+
\sum_{n : \Prob{B_n} = 0} c_n \ind[B_n]{\omega}
+
c_\infty \ind[E]{\omega}
&& (\omega \in \Omega),
\label{eq:conditional-prob-given-sequence}
\end{align}
where $c_n$ and $c_\infty$ are arbitrary constants that cannot be determined
by the sets $(E_n)_{n \ge 1}$ alone.

This representation exposes the nature of neighborhood-limit probabilities:
they pin down only the underdetermined constant
$c_\infty \defeq \Prob{A \given E} = \lim_n \Prob{A \given E_n}$
in a given version of $\Prob{A \given \sigE}$.
The program's meaning under \cref{eq:conditional-probability-via-limits}
rests entirely on an arbitrary choice of $c_\infty$; from the statistical
perspective, however, we can only depend on the behavior of the
$\baseO$-a.s. unique object $\Prob{A \given \sigE}$ under an integral sign.
This semantic issue with neighborhood
limit mirrors the one for density semantics from
\cref{sec:observe-densities}, which depend on values of probability density
functions at specific points.

\subsection{Revisiting the ``Type 1'' Paradox}
\label{sec:neighborhood-bmi}

The neighborhood-limit approach aims to solve the problem of conditioning
on null sets.
We wish to demonstrate that a bare null set is too informal to encode a
meaningful statistical inference problem.
Toward this end, we review the ``Paradox of Type 1'' example~\citep[\S3]{jacobs2021},
where different variables are
observed in different control flow paths (see also \cref{lst:post-control-flow}).

\begin{codeverb}
h = sample normal(1.7, 0.5)  # sample height
w = sample normal(70, 30)    # sample weight
b = sample bernoulli(0.5)
if b == 0
  observe normal(h, 0.1) 2   # noisy measurement of height <@\label{line:bmi-h}@>
else
  observe normal(w, 5) 90    # noisy measurement of weight <@\label{line:bmi-w}@>
bmi = w / h^2
\end{codeverb}

As \citeauthor{jacobs2021} interprets $\observe$
in the sense of statistical conditioning,
the statements on \cref{line:bmi-h,line:bmi-w}
should \textit{not} be understood using the
equivalent $\score$ semantics \cref{eq:equiv-observe-score}.
Rather, \citeauthor{jacobs2021}' discussion states that the program conditions the joint distribution
of $X \defeq (H,W,B,H',W')$ on the null event
\begin{equation}
E \defeq \set{B = 0, H' = 2} \cup \set{B=1, W' = 90},
\label{eq:event-bmi}
\end{equation}
where $H'$ and $W'$ are the implicit random variables used in the two
$\observe$ statements.
\Citeauthor{jacobs2021} shows that using soft conditioning to analyze
this program exhibits \textit{unit anomalies}, in the sense that the
posterior distribution of $X \given E$ changes if we measure height
using centimeters instead of meters.
While acknowledging that some readers may object about conditionally
executed $\observe$, \citeauthor{jacobs2021} argues by analogy that because
they have a rigorous interpretation in the discrete case, they should also
carry good probabilistic meaning in the continuous case.
We study this assertion from the perspective of statistical inference, and
find that unit anomaly problems actually mask a more significant issue that
is caused by interpreting $\observe$ as conditioning on a bare null set.

\subsubsection{Bare Null Events as Observations}

Recall from \cref{prop:limit-not-probability} that
$\Prob{\cdot \given E} = \lim_n \Prob{\cdot \given E_n}$ does not
define a probability distribution, so we require an alternative mathematical
explanation of this limit.
By considering repeated trials of the statistical experiment in the program
above, the null event \cref{eq:event-bmi} must be understood
as a member of a family of events indexed by observable values $z \defeq (z_1,z_2)$
of the noisy height and weight from a given realization of the underlying experiment:
\begin{align}
E(z) \defeq \set{B=0, H'=z_1} \cup \set{B=1, W'=z_2} && \left(z \in \R^2\right),
\label{eq:pseudogynous}
\end{align}
so that $E \equiv E(2, 90)$.
As statistics are by definition measurable
functions of the data (cf.~\zcpageref{quote:schervish}),
we seek to construct a random variable $Z$ in the probability space
such that $E(z) \equiv \set{Z = f(z)}$ for all $z$.
Such a construction, if possible, would lead to a principled statistical inference
problem.
Even though the posterior $X \mid Z=z$ is arbitrary at any fixed $z$
(e.g., the constants 2 and 90 hard-coded in the program), it is perfectly
well-defined as a regular conditional distribution defined earlier.

For events of the form shown in \cref{eq:event-bmi}, such a construction
is generally impossible.
Indeed, there are infinitely many different random variables that encode
$E(z)$ at a \textit{specific} $z$, for example,
\begin{align}
Z' &\defeq \ind[E(z)]{} && E(z) = \set{Z' = 1} \label{eq:bare-null-rv-1}, \\
Z'' &\defeq (1-B)(H'-z_1) + B(W'-z_2) && E(z) = \set{Z'' = 0}. \label{eq:bare-null-rv-2}
\end{align}
But $Z'$ and $Z''$ are of the form $Z^{(z)}$, i.e., the definition of
each random variable depends on the specific observation $z$.
Therefore, each event $E(z) = \set{Z^{(z)} = f(z)}$ in the family
arises from conditioning a \textit{different} random variable $Z^{(z)}$ on a different value.
The individual encodings do not provide a unified encoding using a fixed
observable statistic on the underlying probability space.

\begin{proposition}
Consider the parameterized families of conditioning events $E(z)$ in
\cref{eq:pseudogynous}.
There does not exist a random variable $Z$ and mapping $f$ such that
$E(z) \equiv \set{Z = f(z)}$,
for every $z \in \R^2$
\end{proposition}

\begin{proof}
For $z\defeq(2,90)$ and $z'\defeq(1,90)$, the events $E(z) \cap E(z') = \set{B=1,W'=90} \ne \varnothing$
but $E(z) \ne E(z')$.
If $f(z) \ne f(z')$ then $\set{Z=f(z)} \cap \set{Z=f(z')} = \varnothing$, contradiction.
If $f(z) = f(z')$, then $\set{Z=f(z)} = \set{Z=f(z')}$, so $E(z) = E(z')$, contradiction.
\end{proof}

As a result, we cannot state any consistency property for our statistical
inferences as we change the program constants $(z_1, z_2) = (2, 90)$ to
other values $(z'_1, z'_2)$ that represent
different noisy ``observations'' of the height and weight.
Given only a bare null event $E$, we cannot
formulate a Bayes rule such that
$\pXZ(x\given z) \propto \pX(x)\pZX(z \given x)$
for every setting $z$ of program constants
in the $\observe$ statements.
As \cref{lst:post-control-flow} also conditions on a bare null set,
the possibility of statistical inference in this program is similarly ruled out.
For \cref{lst:post-mixed-type}, this program is admissible, provided that
the $\observe$ measures the \textit{same} statistic $Z$ in both branches.
In this case, using the generative semantics of $\observe$, we can write
a semantically equivalent program that randomizes over the distribution while
holding the observation fixed, as in a standard statistical mixture model:
\begin{codeverb}
x = sample bernoulli(0.5)
if x == 0
  D = uniform(0,1)
else
  D = bernoulli(0.5)
observe D 1
\end{codeverb}

\subsubsection{Ambiguity of Null Set Conditioning}

The argument for neighborhood limits in \citet{jaynes2003} is an attempt
to resolve Borel's paradox: to make the question of
conditioning on a bare null set $E$ meaningful, the user also must specify
additional mathematical structure, such as a limiting process $(E_n)_{n \ge 1}$,
a sigma-algebra of events $\sigE$, or a random variable encoding $Z$.
But this resolution assumes that conditioning on a bare null set is a
statistically sensible problem to begin with.

The representation of $E = \set{B=0, H'=2} \cup \set{B=1,W'=90}$ using
$Z' \defeq 1_E$ underscores the issue.
Since $\Prob{E}=0$, $Z'$ is almost surely equal to $0$.
The regular conditional distribution given
$Z'=0$ is therefore the prior.
But the event of interest is the probability-zero observation $Z'=1$,
for which the defining equation for conditional probability
imposes no constraint whatsoever:
\begin{equation}
  \Prob{X\in A}
  =
  \Prob{X\in A \mid 1_E=1}\Prob{1_E=1}
  +
  \Prob{X\in A \mid 1_E=0}\Prob{1_E=0}.
\end{equation}

Version issues are not pathologies of regular conditional distributions,
because their semantic property is stated by integration against the
marginal law of the observation.
Neighborhood limits, on the other hand, do not define conditional
distributions (\cref{prop:limit-not-probability}) and exhibit more
significant difficulties.
As we showed in \cref{eq:conditional-prob-given-sequence}, they use an
unconstrained value $c_\infty$ at the null fiber as the semantic property
of the entire program.
This choice is external structure that is undetermined by the generative
process, and different equally admissible neighborhoods select different
posteriors.
To fix the ``unit anomaly'' problem,
the method of \citeauthor{jacobs2021}
requires arbitrary constants $w_1, w_2 > 0$
that specify the rate at which the neighborhoods shrink in each program path:
\begin{align}
E_\epsilon \defeq \set{B = 0, H' \in [2-w_1\cdot\epsilon, 2+w_1\cdot\epsilon]} \cup \set{B=1, W' \in [90 - w_2 \cdot\epsilon, 90 + w_2\cdot\epsilon]}.
\end{align}
Let $m_h$ and $m_w$ be the marginal likelihoods
of the noisy observations $\set{H'=2}$ and $\set{W'=90}$, respectively.
In the limit $\epsilon \to 0$, the posterior branch probabilities
depend not only on $m_h$ and $m_w$ (as expected),
but also on the arbitrary constants $w_1$ and $w_2$ selected by the user:
\begin{align}
\Prob{B = 0 \mid E}
&= \lim_{\epsilon \to 0} \Prob{B = 0 \mid E_\epsilon}
= \lim_{\epsilon \to 0} \frac{1/2 \cdot 2 w_1 \epsilon m_h}{1/2\cdot 2 w_1 \epsilon m_h + 1/2 \cdot 2 w_2 \epsilon m_w}
= \frac{w_1 \cdot m_h}{w_1 \cdot m_h + w_2 \cdot m_w}.
\label{eq:neighborhood-branch-arbitrary}
\end{align}

\section{Disintegration Theory}
\label{sec:disintegration}

Our approach to conditioning is based on disintegration, which is
introduced informally in \cref{sec:introduction-disintegration}.
More formally, a probabilistic program operates over a measurable space
$(\X, \SigX)$ of traces with probability measure $\muX$ and $X \sim \muX$ a
random trace.
The observation process is modeled using a measure space
$(\Z, \SigZ)$ and a measurable observation map $\obs: \X \to \Z$.
We define $Z \defeq \obs(X)$ to be the observable random variable,
with marginal distribution $\muZ \defeq \muX \circ \obsI$.
For example, in a standard parametric statistical inference problem,
$X=(\theta,Y)$ is a pair of random variables where $\theta$ is a latent
parameter and $Y$ observable data, with $Z = t(X) = Y$ a projection of the second
coordinate.
Given a specific observation $z \in \Z$, the inference problem
is to infer the conditional distribution $X \mid Z=z$.

A \textit{disintegration} of $X$ through $t$ is a probability kernel
$\muXZ(\cdot \given \cdot): \SigX \times \Z \to [0,1]$ such that
\begin{align}
\muX(A) = \int_{\Z} \muXZ(A\given z) \muZ(\d z)
&& (A \in \SigX),
\label{eq:disintegrate-distX}
\end{align}
where $\muXZ(\cdot \given z)$ is supported on the fiber $\obsI(z)$
for $\muZ$-a.e.~$z$, i.e.,
$\muXZ(\obsI(z) \given z) = 1$ (cf.~\cref{fig:disintegration}).


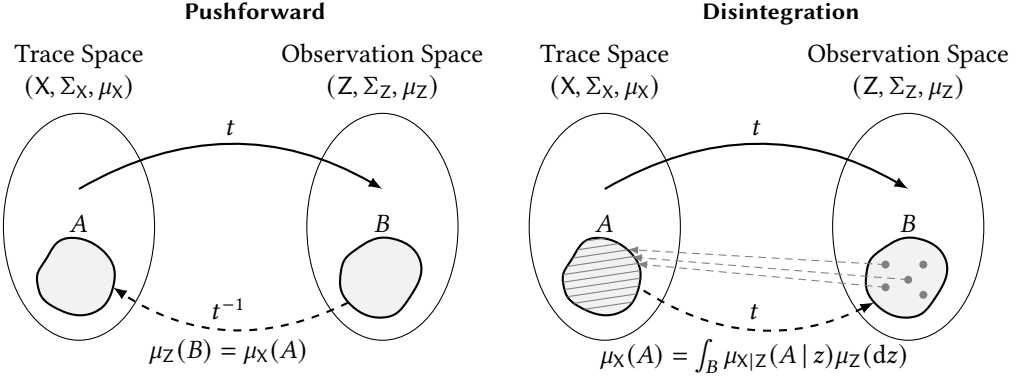
\begin{figure}[!t]
\begin{subfigure}[t]{.5\linewidth}
\captionsetup{font=bf}
\caption*{Pushforward}
\begin{tikzpicture}
\node[name=tX,
  draw, shape=ellipse, minimum height=3cm, minimum width=2cm,
  label={above:\begin{tabular}{c}Trace Space\\ $(\X, \SigX, \muX)$ \end{tabular}},
  ]{};
\node[name=tZ,
  right=2 of tX,
  draw, shape=ellipse, minimum height=3cm, minimum width=2cm,
  label={above:\begin{tabular}{c}Observation Space\\ $(\Z, \SigZ, \muZ)$ \end{tabular}},
  ]{};
\draw[-latex,thick] ([yshift=.5cm]tX.center) to[bend left] node[pos=0.5,above]{$\obs$} ([yshift=.5cm]tZ.center);

\begin{scope}[shift={([yshift=-0.7cm]tX.center)}]
  \draw[thick, fill=gray!10]
    plot[smooth cycle, tension=.8] coordinates {
      (-.45,.25) (-.15,.55) (.35,.35) (.45,-.10) (.15,-.45) (-.35,-.35) (-.55,-.05)
    };
  \node at (0,.75) {$A$};
  \node[name=tA] at (.45,-0.10) {};
\end{scope}

\begin{scope}[shift={([yshift=-0.7cm]tZ.center)}]
  \draw[thick, fill=gray!10]
    plot[smooth cycle, tension=.8] coordinates {
      (-.35,.40) (.10,.55) (.50,.25) (.35,-.25) (.00,-.50) (-.45,-.30) (-.55,.10)
    };
  \node at (0,.75) {$B$};
  \node[name=tB] at (-.45,-.30) {};
\end{scope}

\draw[-latex,thick,dashed] (tB.center) to[bend left]
  node[pos=0.5,below,yshift=-.]{$\muZ(B) = \muX(A)$}
  node[pos=0.5,above]{$\obsI$}
  (tA.center);
\end{tikzpicture}
\end{subfigure}%
\begin{subfigure}[t]{.5\linewidth}
\captionsetup{font=bf}
\caption*{Disintegration}
\begin{tikzpicture}
\node[name=tX,
  draw, shape=ellipse, minimum height=3cm, minimum width=2cm,
  label={above:\begin{tabular}{c}Trace Space\\ $(\X, \SigX, \muX)$ \end{tabular}},
  ]{};
\node[name=tZ,
  right=2 of tX,
  draw, shape=ellipse, minimum height=3cm, minimum width=2cm,
  label={above:\begin{tabular}{c}Observation Space\\ $(\Z, \SigZ, \muZ)$ \end{tabular}},
  ]{};
\draw[-latex,thick] ([yshift=.5cm]tX.center) to[bend left] node[pos=0.5,above]{$\obs$} ([yshift=.5cm]tZ.center);

\begin{scope}[shift={([yshift=-0.7cm]tX.center)}]
  \draw[thick, fill=gray!10]
    plot[smooth cycle, tension=.8] coordinates {
      (-.45,.25) (-.15,.55) (.35,.35) (.45,-.10) (.15,-.45) (-.35,-.35) (-.55,-.05)
    };
  \node at (0,.75) {$A$};
  \node[name=tA] at (.45,-0.10) {};

  \begin{scope}
  \clip
    plot[smooth cycle, tension=.8] coordinates {
      (-.45,.25) (-.15,.55) (.35,.35) (.45,-.10) (.15,-.45) (-.35,-.35) (-.55,-.05)
    };

    \foreach \yy in {-.5, -.4, -.3, -.2, -.1, 0, .1, .2, .3, .4, .5} {
      \draw[gray] (-1,\yy) -- (1,\yy+0.30);
    }
  \end{scope}

  \node[name=in1] at (.3,.4) {};
  \node[name=in2] at (.35,.3) {};
  \node[name=in3] at (.42,.21) {};

\end{scope}

\begin{scope}[shift={([yshift=-0.7cm]tZ.center)}]
  \draw[thick, fill=gray!10]
    plot[smooth cycle, tension=.8] coordinates {
      (-.35,.40) (.10,.55) (.50,.25) (.35,-.25) (.00,-.50) (-.45,-.30) (-.55,.10)
    };
  \node at (0,.75) {$B$};
  \node[name=tB] at (-.45,-.30) {};

  \node[name=out1, fill=gray, circle, inner sep=0pt, minimum size=3pt] at (0,0) {};
  \node[name=out2, fill=gray, circle, inner sep=0pt, minimum size=3pt] at (0.2,0.2) {};
  \node[name=out3, fill=gray, circle, inner sep=0pt, minimum size=3pt] at (0.2,-0.2) {};
  \node[name=out4, fill=gray, circle, inner sep=0pt, minimum size=3pt] at (-0.3,-0.1) {};
  \node[name=out5, fill=gray, circle, inner sep=0pt, minimum size=3pt] at (-0.3,0.2) {};
\end{scope}

\draw[-latex,gray, densely dashed] (out5.center) -- (in1.center);
\draw[-latex,gray, densely dashed] (out1.center) -- (in2.center);
\draw[-latex,gray, densely dashed] (out4.center) -- (in3.center);

\draw[latex-,thick,dashed] (tB.center) to[bend left]
  node[pos=0.5,below]{$\muX(A) = \int_{B}\muXZ(A\given z)\muZ(\d z)$}
  node[pos=0.5,above]{$\obs$}
  (tA.center);
\end{tikzpicture}
\end{subfigure}
\caption{Illustration of pushforward and disintegration of a probability measure
$\muX$ through a map $\obs: \X \to \Z$.
The pushforward operation induces a measure $\muZ$ on $\Z$, where
the measure of a set $B \in \SigZ$ is the measure of its preimage
$A \defeq \obsI(B) \in \SigX$.
Informally, disintegration decomposes the measure of $A$ as the
weighted sum of measures of slices $A \cap \obsI(z) \in \SigX$
that partition $A$,
where $z \in \obs(A)$ is weighted by $\muZ(\d z)$.
These operations give a precise semantics for observations and conditional
distributions in probabilistic programs.}
\label{fig:disintegration}
\end{figure}

\begin{proposition}[store=prop:rcd-disintegration,restate-keys={note={Restated}}]
\label{prop:rcd-disintegration}
A disintegration $\muXZ$
induces a regular conditional distribution
on the probability space $(\X, \SigX, \muX)$
given the sigma algebra $\mathcal{E} \defeq \sigma(t)$,
by setting
$\muX(A \given \mathcal{E})(x) \defeq \muXZ(A \given t(x)).$
\end{proposition}

(\Cref{appx:lemmas,appx:proofs} contain all the proofs from here on.)
Disintegrations exist under a variety of
settings~\citep[\S10.4]{bogachev2007}.
For our purposes, $(\X, \SigX)$ and $(\Z, \SigZ)$ are standard
Borel spaces, which guarantees a
disintegration~\citep[Proposition 10.4.12]{bogachev2007}; \citep[Theorem 1]{chang1997}.
Following the usual setting in statistics, where the observable data $Z \sim \muZ$
we condition is assumed to be drawn from the specified model,
version differences are immaterial
because disintegrations are canonical up to null sets
(i.e., any two versions agree $\muZ$-a.s.).

\subsection{Generalized Bayes Rule}

In our setup, we cannot represent the disintegration $\muXZ$ using the
elementary Bayes rule~\cref{eq:intro-bayes-rule-pdf}, because $Z = \obs(X)$
has a degenerate distribution $\muZX(A \given x) = \dirac{t(x)}[A]$.
Therefore, the density of the ``likelihood'' term $\pZX(z\given x)$
is not yet well specified.
We will construct an explicit representation of the disintegration $\muXZ$
as a probability density $\pXZ$, by way of a more general Bayes rule.

Suppose we have a $\sigma$-finite
base measure $\lamZ$ on the observation space $(\Z, \SigZ)$ that
can be used to disintegrate the base measure $\lamX$ on $(\X,\SigX)$, that is,
there exists a transition kernel $\lamXZ$ such that
\begin{align}
\lamX(A) = \int_{\Z} \lamXZ(A\given z ) \lamZ(\d z) && (A \in \SigX),
\label{eq:disintegrate-baseX}
\end{align}
where $\lamXZ\left(\X \setminus \obsI(z) \given z\right) = 0$
assigns zero measure outside the fiber $\set{x \mid \obs(x) = z}$.
This assumption is sufficient to obtain a generalized Bayes rule representation
of the disintegration.

\begin{theorem}[store=thm:bayes,restate-keys={note={Restated}}]
\label{thm:bayes}
Let $(\X,\SigX,\muX)$ be a probability space with $\muX \ll \lamX$ for a sigma-finite
base measure $\lamX$.
Let $t: (\X,\SigX) \to (\Z,\SigZ)$ be measurable, $\muZ \defeq \pushf{t}{\muX}$,
and $\lamZ$ a sigma-finite measure on $(\Z,\SigZ)$.
If $(\lamZ, \lamXZ)$ is a disintegration of $\lamX$ as in \cref{eq:disintegrate-baseX},
then
\begin{enumerate}[wide=2\parindent, leftmargin=*, label={\zcref*[noname]{thm:bayes}.\arabic*}]

\item \label{thm:bayes-sigma-finite}
  The transition kernel $\lamXZ$ is uniformly $\sigma$-finite, for $\lamZ$-a.e.~$z$.

\item \label{thm:bayes-marginal-z}
  The marginal distribution of $Z$ satisfies $\muZ(\d z) \ll \lamZ(\d z)$, with density
  \begin{equation}
  \pZ(z) \defeq \frac{\d \muZ}{\d \lamZ}(z) = \int_{t^{-1}(z)} \pX(x)\lamXZ(\d x\given z) \in (0, \infty) \mbox{ for } \muZ\mbox{-a.e.}~z.
  \label{eq:prop-bayes-z}
  \end{equation}

\item \label{thm:bayes-posterior-x}
  The conditional distribution of $X$ given $Z$ satisfies $\muXZ(\d x \given z) \ll \lamXZ(\d x \given z)$, with density
  \begin{equation}
  \pXZ(x\given z) \defeq \frac{\d \muXZ(\cdot\given z)}{\d \lamXZ(\cdot \given z)}(x) = \frac{\pX(x)}{\pZ(z)}.
  \label{eq:prop-bayes-x}
  \end{equation}
\end{enumerate}
\end{theorem}

It is essential to select a sigma-finite base measure $\lamZ$ that
satisfies the hypotheses of \cref{thm:bayes}.
For example, the obvious choice $\lamZ = \pushf{\obs}{\lamX}$ as the
pushforward of the trace base measure $\lamX$ through $\obs$ is invalid,
as it may not be sigma-finite.
For example, put $(\X, \SigX, \lamX) = (\R^2, \Borel(\R^2), \Leb^2)$
with $t(x_1,x_2) \defeq x_2$.
Then $\lamZ(B) = \pushf{\obs}{\lamX}(B) = \Leb^2(\R \times B) = \infty$ for
every $B$ with $\Leb(B) > 0$.
Instead, a valid choice in this example is setting $\lamZ = \Leb$ with
$\lamXZ(\d(x_1,x_2)\given z) = \Leb(\d x_1) \otimes \dirac{z}[\d x_2]$.

\section{A Probabilistic Language with Observation Maps}
\label{sec:lang}

\Cref{lst:prob-while-obs} shows a probabilistic language that facilitates
conditioning by disintegration.
Inline $\observe$ statements have been eliminated and a $\return$ statement
followed by a single $\observe$ statement is introduced.
The idea here is that $\return~e$ represents the observation map
$\obs: \X \to \Z$ and $\observe~[r^*]$ fixes the observation $\set{\obs(X) = z}$,
where $z \equiv [r^*] = [r_1, \dots, r_n]$.
The program semantics will not be tied to any specific
values of $\observe~z$, but rather the disintegration property that
establishes correctness for almost-every random data
vector $Z$ in the prior support.

\bgroup
\setlength{\intextsep}{0pt}
\begin{listing}[h]
\captionsetup{belowskip=6pt}
\caption{Language for statistical probabilistic programs with return value.}
\label{lst:prob-while-obs}
\begin{nicebox}[size=minimal,before skip=5pt, after skip=5pt, boxrule=1pt,bottom=0pt,top=0pt]
\setlength{\abovedisplayskip}{0pt}
\begin{align*}
x &\in \mathsf{Vars}\quad
r \in \mathsf{Real}\quad
o \in \mathsf{Op}\quad
D \in \mathsf{Dist}
\\
e &\in \mathsf{Expr} \Coloneq x \mid r \mid [e^*] \mid o(e) \mid D(e) \mid \sample~e  \\
s &\in \mathsf{Stmt} \Coloneq
  x \boldsymbol{=} e
  \mid \mathbf{skip}
  \mid s_1~\textbf{;}~s_2
  \mid \textbf{if}~e~\textbf{then}~s_1~\textbf{else}~s_2
  \mid \textbf{while}~e~\textbf{do}~s
\\
p &\in \mathsf{Prog}~\Coloneq s \bm{\rightsquigarrow} \return~e~\textbf{;}~\observe~[r^*]
\end{align*}
\end{nicebox}
\end{listing}

\egroup

We formulate a disintegration semantics for this language in stages.
In \cref{sec:slp}, we first consider a fixed program path $\pi$ and
identify a local disintegration
\begin{align}
\lamX[\pi](A) = \int_{\Z[\pi]} \lamXZ[\pi](A \given z) \lamZ[\pi](\d z) && (A \in \SigX[\pi]).
\label{eq:disintegration-local}
\end{align}
This base measure will furnish a Bayes rule
$\pXZ[\pi](x \mid z)$ according to \cref{thm:bayes} for an observation
map $\obs[\pi]: \X[\pi] \to \Z[\pi]$ corresponding to the $\return$ statement
associated with that path.

In \cref{sec:mlp}, we construct a global disintegration
for a probabilistic program with multiple paths $\pi \in \Pi$.
This step requires combining the family of local objects from \cref{sec:slp}
into a global disintegration:
\begin{equation}
\setl*{
  (\X[\pi], \SigX[\pi], \lamX[\pi], \muX[\pi]),
  (\Z[\pi], \SigZ[\pi], \lamZ[\pi], \muZ[\pi]),
  (\obs[\pi])
  }_{\pi \in \Pi}
\rightsquigarrow
\setl*{(\X[\Pi], \SigX[\Pi], \lamX[\Pi], \muX[\Pi]),
  (\Z[\Pi], \SigZ[\Pi], \lamZ[\Pi], \muZ[\Pi]),
  (\obs[\Pi])}
\label{eq:meas-space-combine}
\end{equation}
such that
\begin{align}
\lamX[\Pi](A) = \int_{\Z[\Pi]} \lamXZ[\Pi](A \given z) \lamZ[\Pi](\d z) && (A \in \SigX).
\label{eq:disintegration-global}
\end{align}
A key challenge is identifying measure-theoretic structure
that satisfies the hypotheses of \cref{thm:bayes},
to obtain a Bayes rule $\pXZ[\Pi](x \given z)$
for the overall program with observation map $\obs: \X \to \Z$.

\begin{notation}
The symbol $\Haus{d}[m]$ denotes the $d$-dimensional Hausdorff measure on
$\left(\R[m], \Borel[\R[m]]\right)$.
For any measure $\mu$ and event $A$,
we write $\restr{\mu}{A} \defeq E \mapsto \mu(A \cap E)$ to denote the restriction
of $\mu$ to $A$.
If regarded as a measure on
$(A,\restr{\Sigma}{A})$, it is given by
$\restr{\mu}{A}\!(B)\defeq \mu(B)$,
where $\restr{\Sigma}{A}\defeq\set{E\cap A\mid E\in\Sigma}$.
\end{notation}

\subsection{Primitive Probability Distributions}
\label{sec:lang-primitives}

We first formalize the primitive probability distributions in our
calculus.
We fix a countable set $\DistRepo$ of primitive distribution names.
For each $D \in \DistRepo$ there is a distribution family
\begin{equation}
\left(\kindD, \CarrierD, \SigD, \lamD, \pD, \ThetaD, \SigD^\Theta \right),
\label{eq:dist-family}
\end{equation}
where $\kindD \in \set{\kDD, \kDC}$
indicates whether the distribution is discrete or continuous.
\begin{itemize}[nosep]

\item In the discrete case $\kindD = \kDD$,
the carrier $\CarrierD \subset \N$ is a countable set equipped with the
powerset sigma-algebra $\SigD \defeq \Powset{\CarrierD}$ and
counting base measure $\lamD \defeq \Counting[\CarrierD]$.

\item In the continuous case $\kindD = \kDC$,
the carrier $\CarrierD = \R[\rdimD]$ is Euclidean space equipped with the
Borel sigma-algebra $\SigD \defeq \Borel[\R[\rdimD]]$ and
Lebesgue base measure $\lamD \defeq \Leb[\rdimD]$.
\end{itemize}
In both cases, $(\ThetaD, \SigD^\Theta)$ is a measurable parameter space, and the
probability density function $\pD: \CarrierD \times \ThetaD \to [0, \infty)$
is jointly measurable,
such that $\int_{\CarrierD} \pD(y \given \theta) \lamD(\d y) = 1$
for $\theta \in \ThetaD$.
The induced primitive probability distribution is therefore
$\muD(B \given \theta) \defeq \int_{B} \pD(y \given \theta) \lamD(\d y)$
for $B \in \SigD$.

Note that our semantics omit continuous primitive distributions that are
singular with respect to the Lebesgue measure.
Samples from such distributions can be implemented in the language by
sampling from a lower-dimensional coordinate parametrization of the
support.
For example, the uniform distribution over the circle in $\R[2]$ centered at $(0,0)$
with radius $\theta$ is given by
\begin{center}
\lstinline[language=verb]|u = sample uniform(0, 2*pi); x = [theta * cos(u), theta * sin(u)]|
\end{center}
rather than a standalone primitive
\lstinline[language=verb]|x = sample uniform_circle(theta)|.
This assumption is needed to retain a practical language semantics, refer
to \cref{appx:primitives} for a detailed discussion.


\section{Disintegration within a Single Probabilistic Program Path}
\label{sec:slp}

Equipped with the primitives, we construct the disintegration
\cref{eq:disintegration-local} for a single program path $\pi$.

\subsection{Measure Structure}

\paragraph{Path Distributions}
A finite program path $\pi$ contains $\kpath$ dynamic occurrences of $\sample$
statements, which determines an ordered sequence of primitive distribution names
$D_{\pi,1}, \dots, D_{\pi,\kpath} \in \DistRepo$.
For compactness, we omit the subscript $D$ so that
$\kindD[\pi,i] \defeq \kindD[D_{\pi,i}]$ and similarly for remaining quantities
in \cref{eq:dist-family}, where $i=1,\dots,\kpath$.
The path $\pi$ encounters measurable maps
$\theta_{\pi,i}: \CarrierD[\pi,1] \times \dots \times \CarrierD[\pi,i-1] \to \ThetaD[\pi,i]$
that determine the distributional parameters of each $D_i$,
so that for any trace prefix $(y_1, \dots, y_{i-1})$ the concrete distribution
used by the $i$th encountered $\sample$ statement is
$\muD[\pi,i](\cdot \given \theta_{\pi,i}(y_{<i}))$, whose density
with respect to $\lamD[\pi,i]$ is denoted
$\pD[\pi,i](y_i \given y_{<i}) \defeq \pD[D_{\pi,i}](y_i \given \theta_{\pi,i}(y_{<i}))$.

\paragraph{Unconstrained Trace Space}

We first model the traces of $\pi$ ignoring any guard conditions.
The ambient (unconstrained) trace space is given by the product
$(\X[\overline{\pi}], \SigX[\overline{\pi}], \lamX[\overline{\pi}]) \defeq
\otimes_{i=1}^{\kpath}\left( \CarrierD[\pi,i], \SigD[\pi,i], \lamD[\pi,i]  \right)$.
We partition the sample sites into discrete and continuous indexes
\begin{align}
\Idisc \defeq \set*{1 \le i \le \kpath \mid \kindD[\pi,i] = \kDD},
&&
\Icont \defeq \set*{1 \le i \le \kpath \mid \kindD[\pi,i] = \kDC}.
\end{align}
These indexes identify the discrete and continuous (differentiable) structure
in the ambient trace space.
In particular, letting
\begin{align}
\Apath \defeq \textstyle\prod_{i \in \Idisc} \CarrierD[\pi,i], \quad
\dpath \defeq \sum_{i \in \Icont} \rdimD[\pi,i], \quad
\R[\dpath] \cong \textstyle\prod_{i \in \Icont} \CarrierD[\pi,i],
\end{align}
we identify the congruences
\begin{align}
\X[\overline\pi] \cong \Apath \times \R[\dpath],
\quad
\SigX[\overline\pi] \cong \Powset{\Apath} \otimes \Borel[\R[\dpath]],
\quad
\lamX[\overline\pi] \cong \Counting[\Apath] \otimes \Leb[\dpath].
\label{eq:base-meas-congruences}
\end{align}

With this representation, we let $x \defeq (a,w) \in \X[\overline\pi]$
denote a generic trace for $\pi$,
where $a \in \Apath$ contains variables
from the discrete $\sample$ statements and $w \in \R[\dpath]$ the variables
from the continuous $\sample$ statements.
We let $y^\pi_i(a,w)$ extract the $i$th coordinate
in the trace $(a,w)$ in chronological order, i.e., the sample from $D_{\pi,i}$.
The ambient path density with respect to $\lamD[\overline\pi]$ is thus
\begin{align}
\pX[\overline\pi](a,w) = \prod_{i=1}^{k_\pi} \pD[\pi,i]\left(y^\pi_i(a,w) \given y^\pi_{<i}(a,w)\right).
\end{align}

\subsection{Incorporating Path Guards}
\label{sec:slp-guards}

We now refine the measure space to account for guards that
appear in the $\textbf{if}$ and $\textbf{while}$ statements of the program path $\pi$.
Suppose that $\pi$ encounters $\gpath$ guard occurrences.
For each guard occurrence $\ell \in \set{1,\dots,\gpath}$,
let $\kpath[\pi,\ell] \in \set{0, \dots, \kpath}$
be the number of encountered $\sample$ statements
executed before guard $\ell$ is evaluated.
Each guard expression denotes a measurable Boolean function
$\xi_{\pi,\ell}: \CarrierD[\pi,1] \times \dots \CarrierD[\pi,\kpath[\pi, \ell]] \to \set{0,1}$.
The program path fixes a specific Boolean guard value
$b_{\pi,\ell} \in \set{0,1}$, indicating whether the guard evaluates to false or true.
The resulting measure space for $\pi$ is
\begin{gather}
\X[\pi]             \defeq \bigcap_{\ell=1}^{\gpath}\set*{(a,w) \in \X[\overline\pi] \mid \xi_{\pi,\ell}\left(y^{\pi}_{\le k_{\pi,\ell}}(a,w)\right) = b_{\pi,\ell}}, \quad
\SigX[\pi]          = \restr{\SigX[\overline\pi]}{\X_{\pi}}, \quad 
\lamX[\pi]          \defeq \restr{\lamX[\overline\pi]}{\X_{\pi}},
\label{eq:Thomist}
\\
\pX[\pi](a,w)       \defeq \pX[\overline\pi](a,w) \ind{(a,w) \in \X_{\pi}}, \quad
\muX[\pi](\d (a,w)) = \pX[\pi](a,w) \lamX[\pi](\d (a,w)).
\end{gather}

Using the congruences in \cref{eq:base-meas-congruences}, the base
measure---which is evidently sigma-finite using standard counting arguments---is
congruent to
\begin{align}
\lamX[\pi](B)
= \sum_{a \in \Apath} \Leb[\dpath](\B[\pi,a]),
&&
B \in \SigX[\pi], \quad
\B[\pi,a] \defeq \set*{w \in \R[\dpath] \mid (a,w) \in B}.
\label{eq:base-measure-local}
\end{align}
We have thus identified the prior
$\muX[\pi]$ over traces $(a,w)$ within the program path $\pi$
as subprobability distribution.
We assume without loss of generality that
$\muX[\pi](\X[\pi]) \in (0,1]$ is nonzero, because any path
whose guards carve out a $\muX[\pi]$-measure zero subset of $\X[\pi]$
can be semantically ignored.

\subsection{Observation Map}

From \cref{lst:prob-while-obs}, each path $\pi$ ends with
$\textbf{return}~e$, where $e$ is a well-formed expression in the current
path environment that evaluates to a list $[r^*]$ of reals.
Expression $e$ gives a measurable observation map
$\obs[\pi]: (\X[\pi], \SigX[\pi]) \to (\R[\npath], \Borel[\R[\npath]])$.
Note here that a given program $\pi$ always returns a list of fixed length
$\npath$, while other paths may be used to return lists with a different
length.

\paragraph{Continuous Structure}
As $\obs_\pi(x) = \obs_\pi(a,w)$ operates on discrete and continuous coordinates,
we let
$\obs[\pi,a]: \X[\pi,a] \to \R[\npath]$
denote the observation maps on the continuous part,
with $\obs[\pi,a](w) \defeq \obs[\pi](a,w)$.
Here, the sets
$\X[\pi,a] \defeq \set{ w \in \R[\dpath] \mid (a,w) \in \X[\pi] }$
partition the domain $\X[\pi]$
according to the continuous sections of the trace, holding fixed
the discrete values $a \in \Apath$.
A generic path-compatible trace $x \in \X[\pi]$
is written in terms of discrete and continuous coordinates as
$x = (a,w) \in \amalg_{a' \in \Apath} \X[\pi,a']$.

To enforce geometric regularity on the observation maps,
we assume each $\obs[\pi,a]$ is the restriction to $\X[\pi,a]$ of a
$C^1$ (continuously differentiable) constant-rank map on an open domain:
\begin{align}
\obs*[\pi,a]: \U[\pi,a] \to \R[\npath],
\qquad
\obs[\pi,a] = \restr{\obs*[\pi,a]}{\X[\pi,a]},
\end{align}
where $\U[\pi,a] \subset \R[\dpath]$ is open and
$\X[\pi,a] \subset \U[\pi,a]$.
This assumption is mathematically mild, as many observation maps that are not globally
$C^1$ nor constant rank can be handled by stratifying the domain (although
the syntactic implementation requires careful consideration, cf.~\cref{sec:design}).
We define
\begin{equation}
\Z[\pi,a] \defeq \obs*[\pi,a](\U[\pi,a]), \qquad
\Z[\pi] \defeq \cup_{a \in \Apath} \Z[\pi,a] \subset \R[\npath], \qquad
\SigZ[\pi] \defeq \restr{\Borel[\R[\npath]]}{\Z[\pi]}\!.
\end{equation}

\paragraph{Jacobian Matrix}
Transforming a random variable through an observation map requires a
correction for the volume change between the domain and range.
For $w \in \U[\pi,a]$, let
$\obs*[\pi,a]'(w) : \R[\dpath] \to \R[\npath]$
denote the Fr{\'e}chet derivative of the $C^1$ map
$\obs*[\pi,a] = (\obs[\pi,a]^1, \dots, \obs*[\pi,a]^{\npath})$.
In coordinates, we have
\begin{align}
\obs*[\pi]'(x) \defeq
\obs*[\pi,a]'(w) =
  \begin{bmatrix}
    \displaystyle\frac{\partial\obs*[\pi,a]^1}{\partial x_1}(w) &
    \textstyle\frac{\partial\obs*[\pi,a]^1}{\partial x_2}(w) &
    \dots &
    \textstyle\frac{\partial\obs*[\pi,a]^1}{\partial x_{\dpath}}(w)
    \\
    \textstyle\frac{\partial\obs*[\pi,a]^2}{\partial x_1}(w) &
    \textstyle\frac{\partial\obs*[\pi,a]^2}{\partial x_2}(w) &
    \dots &
    \textstyle\frac{\partial\obs*[\pi,a]^2}{\partial x_{\dpath}}(w)
    \\
    \vdots &
    \vdots &
    \ddots &
    \vdots
    \\
    \textstyle\frac{\partial\obs*[\pi,a]^{\npath}}{\partial x_1}(w) &
    \textstyle\frac{\partial\obs*[\pi,a]^{\npath}}{\partial x_2}(w) &
    \dots &
    \textstyle\frac{\partial\obs*[\pi,a]^{\npath}}{\partial x_{\dpath}}(w)
  \end{bmatrix}
  \in \R[ \npath \times \dpath].
  \label{eq:deriv-matrix}
\end{align}
The constant-rank assumption means there exists nonnegative integer
$\mpath \le \min(\dpath, \npath)$ such that
$\rank(\obs*[\pi,a]'(w)) = \mpath$
for all $a \in \Apath$ and
$w \in \U[\pi,a]$.
The Jacobian of the matrix~\cref{eq:deriv-matrix} is defined as
\begin{align}
\Jac[\pi](x) \defeq \Jac(w) = \prod_{i=1}^{\mpath} \sigma_i\left(\obs*[\pi,a]'(w)\right),
\label{eq:coarea-jacobian}
\end{align}
where $\sigma_1 \ge \sigma_2 \ge \dots \ge \sigma_{\mpath} > 0$
extract the nonzero singular values,
with the convention $\Jac(w) = 1$ if $\mpath=0$.
These partial derivatives are often available by standard automatic differentiation
on the $\textbf{return}$ expression of the probabilistic program.
There are two familiar cases for the Jacobian:
\begin{align}
&\mbox{\underline{full column rank} (immersion)}
&
\mpath &= \dpath \le \npath
&
\Jac[\pi](x) &= \sqrt{\det(\obs*[\pi]'(x)^\top \obs*[\pi]'(x))},
\\
&\mbox{\underline{full row rank} (submersion)}
&
\mpath &= \npath \le \dpath
&
\Jac[\pi](x) &= \sqrt{\det \obs*[\pi]'(x) \obs*[\pi]'(x)^\top}.
\end{align}

\subsection{Local Disintegration}
We now have the ingredients to state the main disintegration theorem
for our setup.

\begin{theorem}[store=thm:disintegrate-path,restate-keys={note={Restated}}]
\label{thm:disintegrate-path}
For every $B \in \SigX[\pi]$ and $a \in \Apath$, the Lebesgue measure
disintegrates as
\begin{align}
\Leb[\dpath](\B[\pi,a])
  = \int_{\Z[\pi,a]} \left[
      \int_{\obsI[\pi,a](z)} \frac{\ind[\B[\pi,a]]{w}}{\Jac(w)} \Haus{\dpath-\mpath}[\dpath](\d w)
      \right]
      \Haus{\mpath}[\npath](\d z).
\end{align}
\end{theorem}

This result is based on the coarea formula from \citep[Theorem~3.2.22]{federer1996},
adapted specifically to our topological assumptions on the trace space and observation map.
The disintegration of the Lebesgue measure suggests using the restricted Hausdorff measure
$\restr{\Haus{\mpath}[\npath]}{\Z[\pi]}$
as the observation base measure $\lamZ[\pi]$.
The next results show that this measure is sigma-finite, as needed
in \cref{thm:bayes}.

\begin{theorem}[store=thm:sf-local-base,restate-keys={note={Restated}}]
The measure $\lamZ[\pi] \defeq \restr{\Haus{\mpath}[\npath]}{\Z[\pi]}$ is sigma-finite.
\end{theorem}

Taking the sum over the discrete variables $a \in \Apath$ and applying
Tonelli's theorem to interchange the integral and summation gives the
overall disintegration within the program path $\pi$.

\begin{theorem}[store=thm:disintegration-local,restate-keys={note={Restated}}]
\label{thm:disintegration-local}
The prior base measure $\lamX[\pi]$ over $(\X[\pi], \SigX[\pi])$
disintegrates through $\lamZ[\pi]$ as follows:
\begin{align}
\lamX[\pi](B)           &= \int_{\Z[\pi]} \lamXZ[\pi](B \given z) \lamZ[\pi](\d z)\\
\lamXZ[\pi](B \given z) &\defeq \sum_{a \in \Apath}
                                  \int_{\obsI[\pi,a](z)}
                                    \frac{\ind[\B[\pi,a]]{w}}{\Jac(w)}
                                    \Haus{\dpath-\mpath}[\dpath](\d w),
\qquad
\lamZ[\pi] \defeq \restr{\Haus{\mpath}[\npath]}{\Z[\pi]}\!.
\label{eq:baetyl}
\end{align}
\end{theorem}

\paragraph{Posterior Density}
We next apply \cref{thm:bayes} to represent the disintegration $\muXZ[\pi]$
of $\muX[\pi]/\muX[\pi](\X[\pi])$
using the posterior density $\pXZ[\pi]$.
It is more convenient to use an unweighted
geometric reference measure that dominates the
canonical measure $\lamXZ[\pi]$ identified in \Cref{eq:baetyl}.
For $x = (a,w)$, we set
\begin{align}
\lamXZ*[\pi](B \given z) &\defeq \sum_{a \in \Apath}
                                  \int_{\obsI[\pi,a](z)}
                                    \ind[\B[\pi,a]]{w}
                                    \Haus{\dpath-\mpath}[\dpath](\d w),
&&
\lamXZ[\pi] \ll \lamXZ*[\pi], \quad
\frac{\d \lamXZ[\pi](\cdot \given z)}{\d \lamXZ*[\pi](\cdot \given z)}(x) = \frac{1}{\Jac[\pi](x)}.
\end{align}
Applying the chain rule of Radon--Nikodym derivatives to \Cref{eq:prop-bayes-x} gives
\begin{align}
\pXZ[\pi](x \given z)
\defeq
\frac{\d \muXZ[\pi](\cdot \given z)}{\d \lamXZ[\pi](\cdot \given z)}(x) \thinspace
\frac{\d \lamXZ[\pi](\cdot \given z)}{\d \lamXZ*[\pi](\cdot \given z)}(x)
= \frac{\pX[\pi](x)}{\pZ[\pi](z)\Jac[\pi](x)}.
\label{eq:Helminthocladiaceae}
\end{align}

This expression gives a unified density representation relative to an unweighted
base measure, by moving the Jacobian factor from the base
measure~\cref{eq:baetyl} into the density~\cref{eq:Helminthocladiaceae}.
It can therefore be used to directly compare the
densities of different points $x \defeq (a,w)$ and $x' \defeq (a',w')$
on the fiber $\obsI[\pi](z)$
making it directly compatible with algorithms such as importance sampling and MCMC.


\section{Disintegration of Full Probabilistic Programs}
\label{sec:mlp}

We now construct a disintegration for a probabilistic program
with multiple paths $\pi \in \Pi$, by combining the local
objects~\cref{eq:meas-space-combine} to obtain a global
disintegration~\cref{eq:disintegration-global}.

The global measure space over traces is defined
as the countable coproduct
\begin{gather}
(\X, \SigX, \lamX, \muX) \defeq \textstyle\coprod_{\pi \in \Pi}(\X[\pi], \SigX[\pi], \lamX[\pi], \muX[\pi])
\shortintertext{where the prior density and distribution over traces satisfy}
\pX(\pi, x) \defeq \pX[\pi](x), \qquad
\muX(\d(\pi, x)) \defeq \pX(\pi,x) \lamX(\d(\pi, x)).
\end{gather}
Note that superscripts $\Pi$ have been omitted for clarity.
As the $\return$ statement in each path $\pi$ produces an output in
$\R[\npath]$,
the global observation measurable space is also the coproduct:
\begin{equation}
(\Z, \SigZ) \defeq \textstyle\coprod_{n=0}^{\infty}(\R[n], \Borel[\R[n]]).
\end{equation}
The global observation map $\obs: \X \to \Z$ is given by
$\obs(\pi,x) \defeq \obs[\pi](x)$.
The final step is constructing the global base measure
$\lamZ$ on $(\Z, \SigZ)$, which must dominate all the local path base measures $\lamZ[\pi]$.

\subsection{Observation Base Measure}
\label{sec:mlp-obs}
There are two natural candidates for the observation base measure to consider:
\begin{enumerate*}[label=(\roman*)]
\item the sum $\lamZ[\rm sum] = \sum_{\pi \in \Pi}\lamZ[\pi]$; and
\item the coproduct $\lamZ[\rm cop] \defeq \coprod_{\pi \in \Pi} \lamZ[\pi]$.
\end{enumerate*}
The former, however, may not be sigma-finite.
The latter
is unsuitable because it lives on tagged observations $(\pi,z)$ which
defeats the purpose of posterior inference over the unknown program
path that generated the observation.

We will construct a global base measure $\lamZ$ that avoids both failures.
It is defined on raw observation values and measures each value only
once to ensure sigma-finiteness.
Toward this end, we stratify the observation space $\Z$ according to the output dimension
$\npath$ and observation rank $\mpath$ associated with each path $\pi$.
For each $n \ge 0$ and $m = 0,\dots, n$, we define
\begin{equation}
\Zle{n}{m} \defeq \bigcup_{\pi \mid \npath = n, \mpath \le m} \Z[\pi], \qquad
\Zle{n}{-1} \defeq \varnothing, \qquad
\Zeq{n}{m} \defeq \Zle{n}{m} \setminus \Zle{n}{m-1}.
\label{eq:z-stratify}
\end{equation}
The set $\Zle{n}{m}$ contains observations $z \in \R[n]$ that
lie in the image of an observation map whose rank is at most $m$.
The set $\Zeq{n}{m}$ contains observations $z$ for which
$m$ is the \textit{lowest} such rank.

With these constructions, the observation base measure is defined
as the restriction of the Hausdorff measure to these rank-stratified sets:
\begin{align}
\lamZ(B)
  \defeq
  \sum_{n=0}^{\infty} \sum_{m=0}^{n} \Haus{m}[n][\Zeq{n}{m}]\!(B)
  &&
  (B \in \SigZ).
\label{eq:base-measure-global}
\end{align}

\begin{theorem}[store=thm:lamZ-sigma-finite,restate-keys={note={Restated}}]
\label{thm:lamZ-sigma-finite}
The measure $\lamZ$ over $(\Z, \SigZ)$
is sigma-finite.
\end{theorem}

\subsection{Unifying the Pathwise Base Measures}
\label{sec:mlp-unifying}

We next express all the local base measures
$\lamZ[\pi]$ for program paths $\pi \in \Pi$
in terms of the global base $\lamZ$.
This change-of-measure step is necessary to soundly compare densities that arise
from different program paths with respect to a unified base measure
(recall \Crefrange{eq:z-dens-pi-1}{eq:z-dens-pi-2-bad} from \cref{sec:observe-inferences-mixed}).

\begin{proposition}[store=prop:rd-local-global,restate-keys={note={Restated}}]
\label{prop:rd-local-global}
For each $\pi \in \Pi$,
let $\lamZ[\pi](C) \defeq \Haus{\mpath}[\npath](C \cap \Z[\pi])$
be the extension of the local base measure defined in \Cref{eq:baetyl}
to the global measure space $(\Z, \SigZ)$.
Then $\lamZ[\pi] \ll \lamZ$ and has density
\begin{align}
\rd(z) \defeq \frac{\d \lamZ[\pi]}{\d \lamZ}(z) = \ind{z \in \Z[\pi]} \ind{\rank(z) = \mpath}
&& (\pi \in \Pi, z \in \Z),
\label{eq:zumatic}
\end{align}
where $\rank(z) \defeq \min\set*{ m \mid z \in \Zeq{n}{m}, n \ge 0, m = 0,\dots,n}$,
with the convention $\min\varnothing \defeq +\infty$.
\end{proposition}

This result shows that only those program paths $\pi$ for which the
observation $z$ lies in $\Z[\pi]$ and has \textit{minimal} rank $\mpath$
contribute to the density $\rd(z)$.
Informally, a path whose observation map has rank $m$ assigns a mass of
order $\epsilon^m$ to an $\epsilon$-neighborhood of a regular observation
value, so only the smallest exponents survive in the infinitesimal limit.
The global base measure $\lamZ$ captures this asymptotic comparison.
It measures $z$ only on its minimal rank stratum, while the density
$\rd(z)$ of the local path measures $\lamZ[\pi]$ with respect to $\lamZ$
retains paths whose observation rank is minimal.

The change-of-measure~\cref{eq:zumatic} now lets us identify the global disintegration.

\begin{theorem}[store=thm:disintegration-global,restate-keys={note={Restated}}]
\label{thm:disintegration-global}
The prior base measure $\lamX$ over $(\X, \SigX)$
disintegrates through $\lamZ$ as follows:
\begin{align}
\lamX(B) &= \int_{\Z} \lamXZ(B \given z) \lamZ(\d z),
\qquad
\lamXZ(B \given z) \defeq
  \sum_{\pi \in \Pi} \rd(z) \lamXZ[\pi](\B[\pi] \given z),
\end{align}
for each $B \in \SigX$, where $\B[\pi] \defeq \set{x \in \X[\pi] \mid (\pi,x) \in B}$.
\end{theorem}

\paragraph{Posterior Density}

\Cref{thm:bayes,thm:disintegration-global} together furnish the required
densities $\pZ(z)$ and $\pXZ(\pi,x \given z)$.
As in the case for disintegrating a single path, it is more convenient
for probabilistic inference algorithms to use a conditional density with respect
to an unweighted geometric measure $\lamXZ*$ that excludes the $z$-specific term
$\rd(z)$ and $x$-specific term $\Jac[\pi](x)$:
\begin{equation}
\lamXZ*(B \given z) \defeq \sum_{\pi \in \Pi} \sum_{a \in \Apath} \int_{\obsI[\pi,a](z)} \ind[\B[\pi,a]]{w} \Haus{\dpath-\mpath}[\dpath](\d w),
\qquad
\frac{\d \lamXZ(\cdot \given z) }{ \d \lamXZ*(\cdot \given z) }(\pi, x)
= \frac{\rd(z)}{\Jac[\pi](x)}.
\end{equation}

Analogously to \cref{eq:Helminthocladiaceae},
applying \cref{thm:bayes} to the disintegration in \cref{thm:disintegration-global}
and then changing the base measure from
$\lamXZ$ to $\lamXZ*$ gives the overall conditional density:
\begin{align}
\pXZ(\pi, x \given z)
\defeq
\frac{\d \muXZ(\cdot \given z)}{\d \lamXZ(\cdot \given z)}(\pi,x) \thinspace
\frac{\d \lamXZ(\cdot \given z)}{\d \lamXZ*(\cdot \given z)}(\pi,x)
= \frac{\pX(\pi, x)\rd(z)}{\pZ(z)\Jac[\pi](x)}
= \frac{\pX[\pi](a,w)\rd(z)}{\pZ(z)\Jac(w)},
\label{eq:Helminthocladiaceae-2}
\end{align}
for any program path $\pi \in \Pi$ and trace $x = (a,w) \in \amalg_{a \in \Apath[\pi]}\X[\pi,a]$.

\subsection{Relationship to Neighborhood Limits}
\label{sec:mlp-neighborhood}

The disintegration $\muXZ$ that we have constructed admits
a neighborhood-limit interpretation.
\Citet{pfanzagl1979} shows that under standard topological assumptions
(satisfied by our construction), for a fixed observation map
$\obs: \X \to \Z$, a conditional distribution can be obtained as a
derivative of measures on $\Z$.
For a neighborhood $C$ of $z$ such that $\muZ(C)>0$, define the measure
\begin{align}
M^z_C(B) \defeq \frac{\muX(B \cap \obsI(C))}{\muX(\obsI(C))} && (B \in \SigX).
\label{eq:pfanzagl}
\end{align}
Informally, \citeauthor{pfanzagl1979} proves that $M^z_C(\cdot)$ converges to a regular
conditional distribution $\muXZ(\cdot \given z)$ as
$C \to z$, for $\muZ$-almost every $z$.
The object is unique up to $\muZ$-null sets.
Different admissible
neighborhood bases for the same map $\obs$ only change the
resulting version of the regular conditional distribution.
This result does not contradict \cref{prop:limit-not-probability}, because
\citeauthor{pfanzagl1979} takes the \textit{weak} limits of probability measures
$M^z_C$, not the setwise limits $B \mapsto \lim_C M^z_C(B)$
that \citeauthor{jaynes2003} advocates for.

Neighborhood limits are therefore sound when they arise from a coherent
observation map.
The neighborhood-limit approach of \citet{jacobs2021} from \cref{sec:observe}
does not guarantee this condition, because it uses branch-dependent neighborhood
approximations such as $\cup_{\pi} \set{\pi} \times \obsI[\pi](C_{\pi,\epsilon})$,
which may not be the preimage of a single neighborhood under a fixed observation map.
The limiting branch probabilities can therefore contain arbitrary scale factors
from the chosen neighborhoods, as shown in \Cref{eq:neighborhood-branch-arbitrary}.
Our global observation map $\obs: \X \to \Z$ and base measure $\lamZ$ rule out extra
choices exogenous to the statistical model, because all program
paths are compared through the same observation-space derivative.
Our construction of $\muXZ$ explicitly represents the
weak limit of the measure-theoretic derivative~\cref{eq:pfanzagl} for probabilistic program
traces as a posterior density $\pXZ$~\cref{eq:Helminthocladiaceae-2}.

\subsection{Projection Observations and Likelihood Weighting}
\label{sec:mlp-lw}

The factor $r_\pi(z)$ in the posterior density $\pXZ$ is a global program
property that requires knowing \begin{enumerate*}[label=(\roman*)]
\item whether $z$ lies in the path image; and
\item whether the path has minimal rank among \textit{all} paths that can generate $z$.
\end{enumerate*}
In the special case of projection observations, this factor can be soundly
estimated by sampling, and we can prove that inference via likelihood
weighting is strongly consistent.

Suppose the language in \cref{lst:prob-while-obs} is equipped with
a deterministic addressing scheme for $\sample$ statements.
Each execution along a program path $\pi$ produces
an \textit{addressed} trace
$x=\set{(\alpha, x_\alpha)}_{\alpha\in\mathcal A_\pi}$,
where $\mathcal A_\pi \subset \N$ is the set of encountered addresses.
The observation maps are coordinate projections:
\begin{align}
\obs[\pi](x) = x_{O_\pi} \defeq \set{(\alpha, x_\alpha)}_{\alpha \in O_\pi},
\qquad
O_\pi \subset \mathcal A_\pi
&& (\pi \in \Pi).
\end{align}
With the identification of addressed traces as Euclidean vectors,
the Jacobian
$\Jac[\pi](x)\equiv 1$ and the path rank $\mpath[\pi]$ is the number of
continuous coordinates in $O_\pi$.
The nonlocal term in $\pXZ$ is $\rank(z)$:
\begin{equation}
\pXZ(\pi, x \given z) \propto \pX[\pi](x) \ind{\obs[\pi](x)=z} \ind{m_\pi=\rank(z)} .
\label{eq:projection-target-weight}
\end{equation}

Let $\QXZ$ be the usual likelihood-weighting proposal over $\X$ for an observation
$z$, which generates a trace $x$ by sampling all unobserved
coordinates $x_{\mathcal{A}_\pi \setminus O_\pi}$ from their prior conditionals
and fixing observed coordinates $x_{O_\pi}$ to $z$.
The likelihood weight $L$ is the product of the probability
densities (or masses) of the fixed coordinates and $\ind{\obs[\pi](x)=z}$.
Draw $N>0$ i.i.d.~importance-weighted samples
\begin{equation}
\left(\left(\left(\bm\pi_j,X_j\right), L_j\right)\right)_{j=1}^N\sim \QXZ[N].
\end{equation}
As $\rank(z)$ is unknown, it is estimated empirically from samples as follows:
\begin{align}
\rank_N(z) \defeq \min\set{M_j \mid 1\le j\le N,\ L_j>0},
&&
M_j \defeq \mpath[\bm\pi_j]\quad
(1 \le j \le N).
\label{eq:empirical-min-rank}
\end{align}
Finally, for a measurable function $f:\X \to \R$, define the likelihood-weighting
ratio estimator
\begin{equation}
\widehat{\mathbb E}^z_N[f]
\defeq
\frac
  {\sum_{j=1}^N f(\bm\pi_j,X_j)L_j\,\ind{M_j=\rank_N(z)}}
  {\sum_{j=1}^N L_j\,\ind{M_j=\rank_N(z)}},
\label{eq:ratio-est-lw}
\end{equation}
which is set arbitrarily when the denominator is zero.

\begin{theorem}[store=thm:importance-sampling,restate-keys={note={Restated}}]
\label{thm:importance-sampling}
Let $f: (\X,\SigX) \to (\R,\Borel[\R])$ be bounded and measurable, and
$z$ be an observation with $\pZ(z)\in(0,\infty)$.
Likelihood weighting is strongly consistent for
$\Expect[\muXZ(\cdot \given z)]{f(\bm\pi,X)}$, that is,
\begin{align}
\lim_{N \to \infty} \widehat{\mathbb E}^z_N[f] = \int_{\X} f(\pi,x)\,\muXZ(\d(\pi,x)\mid z),
&&
\QXZ[\infty]\mbox{-a.s.}
\end{align}
\end{theorem}

The key idea is that with projection observations, likelihood weighting
does not need to know the global rank statically.
The event $\set{L>0,\ M=\rank(z)}$ has positive $\QXZ$-probability, so by
Borel--Cantelli it occurs infinitely often with $\QXZ[\infty]$-probability one.
The empirical estimator $\rank_N(z)$ will therefore eventually equal $\rank(z)$,
after which the likelihood-weighting ratio estimator reduces to ordinary
self-normalized importance sampling within the program paths of minimal-rank.
A similar argument can be used for general proposal importance-sampling
distributions $\QXZ$ (not necessarily the prior) and for randomized addressing
schemes from a countable set.


\section{Language Design Considerations}
\label{sec:design}

To our knowledge, no existing probabilistic programming language supports
in full generality the geometric semantics for conditioning developed
in this paper.
We expect implementations to specialize our framework in different ways,
depending on the desired tradeoff between expressiveness, runtime
performance, and static verification.
We briefly discuss implementation considerations that arise when designing
systems whose inference algorithms leverage our Bayes rule:
\begin{align}
\pXZ(\pi, x \given z) = \frac{\pX[\pi](x)\rd(z)}{\pZ(z)\Jac(w)} && (x = (a,w)).
\label{eq:multiganglionic}
\end{align}
The easy terms are $\pX[\pi](x)$, which
is available by executing the program on a fixed trace,
and $\pZ(z)$, which can be ignored by inference algorithms that
do not depend on normalizing constants.

\paragraph{Guards}
\Cref{thm:disintegrate-path} allows the path domain $\X[\pi]$ to be
an arbitrary positive-measure subset of the ambient Euclidean
trace space carved out by program guards (\cref{sec:slp-guards}).
Even exotic domains such as a fat Cantor set, which has positive
Lebesgue measure but empty interior, are admissible in the mathematical
theory.
In practice, however, language implementations are likely to restrict
guards to more structured classes, such as expressions built from
piecewise-analytic-under-analytic-partition (PAP) functions~\citep{lee2020}, in
order to make path analysis and inference tractable.

\paragraph{Observation Map}
Enforcing that each path-specific observation map $\obs[\pi]: \X[\pi] \to \Z[\pi]$ is
constant-rank may require explicit stratification using program
branches.
For example, the map $t(x) = x^2$ has rank 0 at $x=0$ and rank 1 at
$x \ne 0$; if $x$ is a mixed-type random variable (e.g., a normal-Bernoulli
mixture), then explicit paths are needed for the $x=0$, $x<0$, and $x>0$ cases.
Stratification is theoretically straightforward but may be nontrivial
to expose in program syntax, and may require support from the host
language to detect and construct valid constant-rank observation maps.

\paragraph{Jacobian}

Because the observation maps in the main theory are restrictions of $C^1$
maps on Euclidean domains, the Jacobian $\Jac[\pi](x) = \Jac[\pi,a](w)$ is
computed from ordinary Euclidean partial derivatives.
These derivatives can often be obtained by automatic differentiation of the
return expression, once the trace $x=(a,w)$ is typed into discrete and
continuous coordinates.
A practical AD implementation requires enough type information to
distinguish the discrete and continuous coordinates, so that $a$ is held
constant and differentiation is performed with respect to $w$.

\paragraph{Fiber Support}
The posterior distribution $\muXZ$
is always supported on the fiber $\obsI(z)$, so the density term $\pXZ$ \cref{eq:multiganglionic}
implicitly includes an indicator $\ind{\obs(\pi, x) = z}$.
Generating samples on this fiber---such as for proposal distributions in
importance sampling or Markov chain Monte Carlo---is a key challenge.
If observation maps are projections of sampled
variables, the fiber is a coordinate slice, so importance sampling
using the prior works as a default proposal (\cref{sec:mlp-lw}).
In general, however, the fiber may be a complex set,
and there is no default proposal strategy.
Practical language implementations will likely require user guidance,
specialized constraint solvers, reparameterizations of the fiber, or
programmable inference kernels that can effectively explore the posterior
geometry.

\paragraph{Rank}
\Cref{sec:mlp-lw} showed a special case where the global rank information
can be consistently estimated via importance sampling.
In restricted languages, observation ranks may be
obtained by exploiting known structure of the probabilistic model.
In more flexible languages, a combination of program analysis and
statistical estimation methods may be required to compute the ranks.


\section{Related Work}
\label{sec:related}

\subsection{Geometric Densities and Base Measures}

\Citet{diaconis2013sampling}
leverage geometric techniques to solve statistical problems that involve
probability distributions on Euclidean submanifolds.
This perspective is closely related to our construction in \cref{sec:slp},
where the coarea formula disintegrates the Lebesgue measure along the
fibers of a path-specific observation map.
Our contribution differs in scope, however.
Rather than sampling from a \textit{fixed} manifold, we construct in
\cref{sec:mlp} disintegration semantics for entire probabilistic program
traces.
Our theory supports constructs beyond those in
\citet{diaconis2013sampling}, such as countably many execution paths,
discrete-continuous mixtures, overlapping
path images, and a global sigma-finite observation base measure that
gives a Bayes rule over probabilistic program traces.

\Citet{radul2021base} develop a powerful calculus for computing densities
of probability distributions on manifolds.
Using our notation, their construction can be understood as a principled
representation of the marginal density $\pZ$ over the observation space
$\Z = \cup_{n=0}^{\infty} \R[n]$.
The essential difference between our developments lies in the semantic
objects of interest and resulting geometric machinery.
Our task goes further, where a
probabilistic program denotes a prior measure $\muX$ over traces together
with an observation map $\obs: \X \to \Z$, and conditioning is defined by
disintegration:
\begin{equation}
\muX(A)
= \int_{\Z} \muXZ(A \given z) \muZ(\d z)
= \int_{\Z} \int_A \pXZ(\pi, x \given z) \lamXZ(\d(x,\pi) \given z) \pZ(z) \lamZ(\d z).
\end{equation}
As we aim to construct a conditional law $\muXZ$ and density
$\pXZ$ over path-indexed traces $(\pi,x)$,
the observation density $\pZ$ falls out only as the normalizing constant.
This distinction requires us to develop alternative measure-theoretic
constructions to the ones in \citeauthor{radul2021base} that enable a
disintegration theorem and Bayes rule.
These differences are most visible in three places.

\begin{itemize}[wide=0pt]

\item \textit{Countably Infinite Path Supports}.
\Citet[\S4.2]{radul2021base} work with a \textit{finite} union of
manifolds and take the sum of Hausdorff measures as the common base.
Probabilistic programs may involve countably many manifolds where the sum
of base measures is not appropriate (cf.~\cref{sec:mlp-obs}).
We construct an alternative base measure $\lamZ$
that lets us apply \cref{thm:bayes} to derive a Bayes rule.

\item \textit{Positive Measure Overlaps}.
For two manifolds $M_1$ and $M_2$ in $\Z$ with dimensions $d_1$ and $d_2$,
\citeauthor{radul2021base} assume that the intersection $M_1 \cap M_2$
has a lower dimension than $\max(d_1, d_2)$.
This assumption is fine for representing a standalone measure or
density on $\Z$, but is not compatible with path-based disintegration.
Observation maps from different program paths can overlap on a
positive-measure subset of the same rank, and this overlap is precisely where
posterior uncertainty about the branching structure of the program lives.
%
%
In our theory, these overlaps contribute additively because the
global kernel
$\lamXZ(B \given z) = \sum_{\pi\in\Pi} \rd(z)\lamXZ[\pi](\B[\pi]\given z)$
retains every minimal-rank path that can produce $z$.
Each observation value is therefore measured once by $\lamZ$, while posterior
uncertainty over the generating path is represented on the trace side by
the density $\pXZ((\pi,x)\given z)$.

\item \textit{Observation Maps}.
The change-of-variables result in \citet[Theorem 1]{radul2021base}
applies to a diffeomorphic map on manifolds, where each
point has a unique preimage.
Our theory handles noninjective maps, whose images are not necessarily
single embedded manifolds but countably rectifiable (\cref{def:rectifiable}).
Noninjective constant-rank maps include immersions, which can have multiple
zero-dimensional preimages, and submersions, which can have a
positive-dimensional fiber.
To handle such observation maps we derive a coarea formula
(\cref{thm:disintegrate-path}) to disintegrate the trace base measure over
$\obsI(z)$, with a Jacobian correction for the volume change.
\end{itemize}

\vspace{-.25cm}

\subsection{Probabilistic Programming Systems with Disintegration}

On the implementation side,
\citet{narayanan2016,shan2017,narayanan2020symbolic} develop sophisticated
inference techniques in Hakaru that produce disintegrations symbolically
via semantics-preserving program transformations.
These systems are largely formalized using operational rewrite rules,
whereas our contribution provides a mathematical foundation for
disintegrating probabilistic program traces at the measure-theoretic level.
In addition, several modeling constructs handled by our theory are not
handled by existing solver-based implementations.

Certain systems implement the restricted form of disintegration from \cref{sec:mlp-lw},
where the observation map is a projection.
\Citet{lew2020trace} adopt this formalism and build a type system that
rules out mixed-type distributions and other constructs where density
semantics fail.
\Citet{wu2018discrete} and \citet{saad2021sppl} develop generalized
density semantics and probabilistic inference algorithms that avoid
pitfalls of soft conditioning and correctly handle mixed-type distributions
for restricted probabilistic languages that specify Bayesian networks and
sum-product networks, respectively.

\subsection{Language Semantics and Disintegration}

Quasi-Borel spaces~\citep{heunen2017} provide a categorical foundation for
higher-order probability, replacing measurable spaces with a Cartesian
closed setting that supports function types.
This direction is complementary to our work.
We operate in a standard Borel space of trace and observation spaces, and
construct an explicit geometric disintegration of the measure on traces.

\Citet{dahlqvist2020} give a denotational semantics for
probabilistic programs using linear operators on Banach spaces,
where conditioning is modeled as the adjoint of the forward generative operator.
\citep[Remark~25]{dahlqvist2020}.
In contrast to this work, we construct explicit posterior kernels and
densities for exact observations in a space of path-indexed traces, which
can be used directly by probabilistic inference algorithms such as
importance sampling and MCMC.

\Citet{vakar2026} study s-finite kernels as a semantic domain
for probabilistic programs and prove s-finite analogues of Radon--Nikodym
and disintegration theorems.
For our goal of an explicit Bayes rule, an s-finite observation base
measure is not sufficient.
As we have shown, the sum of local pathwise base measures can be s-finite but not
sigma-finite, which makes it challenging to compute densities.
Therefore, our construction specifically uses a sigma-finite base measure,
giving standard Radon--Nikodym densities and a Bayes rule representation of
the disintegration.

\section*{Acknowledgements}

This material is based upon work supported by the National Science Foundation
under Award No.~2311983.
Any opinions, findings and conclusions or recommendations expressed in this
material are those of the authors and do not necessarily reflect the views
of the National Science Foundation.

The author gratefully acknowledges Jesse M.~Michel, whose request for a
coherent mathematical account of conditioning in probabilistic programs
served as an impetus for the present work.

\bibliographystyle{ACM-Reference-Format}
\bibliography{paper}

\AtEndDocument{
  \appendix
\section{Technical Lemmas}
\label{appx:lemmas}

\begin{notation}
\label{notation:measure-space-prod-coprod}
Let $I$ be a countable index set and
$\langle(\X_i, \Sigma_i, \lambda_i)\rangle_{i \in I}$ be a family
of measure spaces.
The product measure space
$(\X, \Sigma, \lambda) = \otimes_{i \in I}(\X_i, \Sigma_i, \lambda_i)$
is defined by
\begin{align}
\X &= \prod_{i \in I} \X_i,
&
\Sigma &\defeq \sigma\left( \set*{\prod_{i \in I} E_i \;\middle|\; E_i \in \Sigma_i}  \right),
&
\lambda\left(\prod_{i \in I}E_i\right) &\defeq \prod_{i \in I} \lambda_i(E_i).
\shortintertext{
In the product construction, if $I$ is infinite then the construction
assumes $\lambda_i(X_i) \in [0,1]$.
The coproduct (disjoint union) measure space
$(\X, \Sigma, \lambda) = \coprod_{i \in I}(\X_i, \Sigma_i, \lambda_i)$
is defined by}
\X &= \coprod_{i \in I} \X_i,
&
\Sigma &\defeq \set*{\coprod_{i \in I} E_i \;\middle|\; E_i \in \Sigma_i},
&
\lambda\left(\coprod_{i \in I}E_i\right) &\defeq \sum_{i \in I} \lambda_i(E_i).
\end{align}
In the coproduct construction, if the sets $\X_i$ are disjoint, we may
implicitly take the disjoint union symbol $\amalg$ to mean the regular
(untagged) union of sets, otherwise it denotes the tagged union.
\end{notation}

\begin{definition}
\label{def:rectifiable}
A set $M \in \Borel(\R^m)$
is called countably $d$-rectifiable if there exists a family
\begin{align}
f_n : A_n \to \R^m && (A_n \in \Borel(\R^d), n \ge 1)
\label{eq:rectifiable}
\end{align}
of Lipschitz functions
and a set $M_0 \subset M$ such that
\begin{align}
\Haus{d}[m][M]\!(M_0) = 0 && M \subset M_0 \cup \bigcup_{n=1}^{\infty} f_n(A_n).
\end{align}
In particular, (almost) all of $M$ can be traced out by only $d$ parameters.
\end{definition}

\begin{remark}
If $M$ is a $C^1$ embedded manifold of dimension $d$, then
$M$ is countably $d$-rectifiable with $\Haus{d}[m](M) > 0$.
\end{remark}

\begin{lemma}
\label{lemma:haus-restrict-rectifiable}
If $M \in \Borel[\R^m]$ is countably $d$-rectifiable, then
$\Haus{d}[m][M]$ is $\sigma$-finite.
\end{lemma}

\begin{proof}
Decompose each set $A_n \subset \R^d$
from \cref{eq:rectifiable} into bounded subsets
$A_{n,k} \defeq A_n \cap [-k,k]^{d}$ for $k \ge 1$.
As $A_{n,k}$ is Borel and $f_n$ is continuous, the image
$f_n(A_{n,k})$ is analytic, and thus Hausdorff measurable.
It also has finite measure from the Lipschitz bound
\begin{equation}
\Haus{d}[d](A_{n,k}) \le \Haus{d}[d]\left([-k,k]^d\right) = (2k)^d < \infty.
\end{equation}
Let $K_n$ be the Lipschitz constant for $f_n$.
The Hausdorff-measure bound for Lipschitz maps~\cite[Corollary 2.10.11]{federer1996} gives
\begin{equation}
\Haus{d}[m](f_n(A_{n,k})) \le K^d_n \cdot \Haus{d}[d](A_{n,k}) < \infty.
\end{equation}
As $A_n = \cup_{k=1}^{\infty}A_{n,k}$,
by Borel regularity of $\Haus{d}[m]$ we may enlarge $M_0$ and each
$f_n(A_{n,k})$ to Borel supersets of the same measure.
Intersecting these sets with $M$ gives a countable Borel cover of $M$
by sets of finite $\Haus{d}[m][M]$ measure.
\end{proof}

\begin{lemma}
\label{lemma:rectifiable-haus-zero}
If $M \in \Borel[\R^m]$ is countably $d$-rectifiable, then
$\Haus{d'}[m](M) = 0$ for all $d' > d$.
\end{lemma}

\begin{proof}
By \cref{lemma:haus-restrict-rectifiable}, $\Haus{d}[m][M]$ is sigma-finite,
so there exists a cover $\cup_{j=1}^{\infty} M_j = M$
such that $M_j \subset M$ for each $j$ and $\Haus{d}[m](M_j) < \infty$.
Therefore, $\Haus{d'}[m](M_j) = 0$ for all $d' > d$
\citep[2.10.2(1)]{federer1996}.
By countable additivity, we have
$\Haus{d'}[m](M) \le \sum_{j=1}^{\infty} \Haus{d'}[m](M_j) = 0$.
\end{proof}

\begin{lemma}
\label{lemma:c1-borel-rectifiable}
Let $U \subset \R[d]$ be open and $f: U \to \R[n]$ be a $C^1$ function.
If there exists $m \le \min(d,n)$ such that $\rank(f'(x))=m$ for all $x \in U$,
then $f(U) \subset \R[n]$ is Borel and countably $m$-rectifiable.
\end{lemma}

\begin{proof}
By the constant-rank theorem, for each $x \in U$ there exists an open neighborhood
$V_x \subset U$ such that $f(V_x)$
is a $C^1$ embedded $m$-dimensional submanifold of $\R[n]$,
in particular it is Borel and countably $m$-rectifiable.
As $U$ has a countable base, choose points $x_j \in U$, $j\ge 1$ such that
$(V_{x_j})_{j \ge 1}$ covers $U$.
Then $f(U) = f(\cup_{j=1}^{\infty}V_{x_j}) = \cup_{j=1}^{\infty} f(V_{x_j})$.
As the union is countable, the conclusion follows.
\end{proof}

\begin{lemma}
\label{lemma:coarea-c1}
Let $\X \in \Borel[\R[d]]$ be measurable,
let $U \subseteq \R[d]$ be open with $\X\subseteq U$, and
let $\bar f:U\to \R[n]$ be $C^1$.
Suppose there is $m\le \min(d,n)$ such that
$\rank(\bar f'(w))=m$ for all $w\in U$.
Let
$f \defeq \restr{\bar f}{\X}$,
and
$\Z \defeq \bar{f}(U)$.
Define the Jacobian
$\Jac[f](w) \defeq \prod_{i=1}^{m}\sigma_i(\bar{f}'(w))$
as the product of the first $m$ nonzero singular values of $\bar{f}'(w)$,
with the convention $\Jac[f](w)=1$ when $m=0$.
Then for every
$B\in \restr{\Borel[\R[d]]}{\X}$,
\begin{equation}
\Leb[d](B) =
  \int_{\Z} \left[
      \int_{f^{-1}(z)} \frac{\ind[B]{w}}{\Jac[f](w)} \,\Haus{d-m}[d](\d w)
    \right] \Haus{m}[n](\d z).
\end{equation}
\end{lemma}

\begin{proof}
This result follows from the coarea formula
in \citet[Theorem 3.2.22]{federer1996}.
As the coarea formula is stated for Lipschitz maps, we establish
the result for our assumption that $\bar{f}$ is $C^1$.

By \cref{lemma:c1-borel-rectifiable}, $\Z= \bar{f}(U)$ is Borel and countably
$m$-rectifiable.
Choose compact convex sets $K_j\subseteq U$, $j\ge 1$, with
$U=\bigcup_{j=1}^{\infty}K_j$,
and define the Borel partition
\begin{equation}
E_j \defeq \X \bigcap \left(  K_j\setminus\bigcup_{\ell=1}^{j-1} K_\ell \right).
\end{equation}
Since $\bar{f}$ is $C^1$, its derivative is bounded on each $K_j$.
Thus $\restr{f}{E_j}:E_j\to \Z$ is Lipschitz.

Each $E_j$ is Borel, hence countably $d$-rectifiable and
$\Haus{d}[d]$-measurable.
Therefore the coarea formula
from \citet[Theorem 3.2.22]{federer1996} applies to
$\restr{f}{E_j}:E_j\to \Z$.
In particular, by the measurability assertion in the coarea formula,
for every nonnegative Borel function $g:\X\to[0,\infty]$,
the map
\begin{equation}
z\longmapsto \int_{f^{-1}(z)\cap E_j} g(w)\,\Haus{d-m}[d](\d w)
\end{equation}
is Borel measurable.
Indeed, for every compact set $C\subset U$,
the map
$z\mapsto\Haus{d-m}[d](C\cap\bar f^{-1}(z))$
is Borel~\citep[Theorem~6.1]{MattilaMauldin1997}.
These measures are finite on compact subsets of $U$ by the
constant-rank assumption.
A standard monotone-class argument on
each $K_j$ extends the conclusion to every nonnegative
Borel integrand.

In our $C^1$ setting, the approximate Jacobian in \citeauthor{federer1996}
agrees with the ordinary Jacobian $\Jac[f]$.
Applying coarea on $E_j$ with the function
\begin{equation}
g_j(w) \defeq \frac{\ind[B\cap E_j]{w}}{\Jac[f](w)}
\end{equation}
gives
\begin{equation}
\Haus{d}[d](B\cap E_j)
     = \int_Z \left[
          \int_{f^{-1}(z)}
          \frac{\ind[B\cap E_j]{w}}
               {\Jac[f](w)}
          \Haus{d-m}[d](\d w)
     \right]
     \Haus{m}[n](\d z).
\end{equation}
Summing over $j$ and applying Tonelli's theorem yields
\begin{equation}
\Haus{d}[d](B) = \int_Z \left[
    \int_{f^{-1}(z)} \frac{\mathbf 1_B(w)} {\Jac[f](w)}
    \Haus{d-m}[d](\d w)
    \right]
  \Haus{m}[n](\d z).
\end{equation}
The Hausdorff measure $\Haus{d}[d]$ agrees with the Lebesgue
$\Leb[d]$ on $\R[d]$ by the standard normalization.
\end{proof}

\section{Deferred Proofs}
\label{appx:proofs}

\getkeytheorem{prop:rcd-disintegration}

\begin{proof}
For $B \in \sigma(t)$, write $B = \obsI(C)$ for some $C \in \SigZ$.
Then, for any $A \in \SigX$, we have
\begin{equation}
\int_{B} \muX(A \given \sigE{})(x) \muX(\d x)
\defeq
\int_{B} \muXZ(A \given t(x)) \muX(\d x)
=
\int_{C} \muXZ(A \given z) \muZ(\d z)
=
\muX(A \cap \obsI(C)).
\end{equation}
The first equality follows by definition and the second by the fact that
$\muZ = \pushf{t}{\muX}$.
The equality follows from the fiber-support condition
$\muXZ(\obsI(z) \given z) = 1$, which implies that for fixed $z \in \Z$,
\begin{equation}
\muXZ(A \cap \obsI(C) \given z) = \begin{cases}
  \muXZ(A \given z) & \mbox{if } z \in C \\
  0 &\mbox{if } z \notin C
  \end{cases}
\;
=
\ind[C]{z} \muXZ(A \given z).
\end{equation}
Whence
\begin{equation}
\muX(A \cap \obsI(C))
= \int_{\Z} \muXZ(A \cap \obsI(C) \given z) \muZ(\d z)
= \int_{\Z} \ind[C]{z} \muXZ(A \given z) \muZ(\d z).
\end{equation}
\end{proof}

\getkeytheorem{thm:bayes}

\begin{proof}

To establish \zcref{thm:bayes-sigma-finite}, we show there exists
a countable cover $\X = \cup_{n=1}^{\infty} K_n$ of events $K_n \in \SigX$,
independent of $z$, such that
\begin{align}
\lamXZ(K_n \given z) < \infty && (\mbox{for }\lamZ\mbox{-a.e.} z \mbox{ and all } n \ge 1).
\end{align}
Since $\lamX$ is $\sigma$-finite,
let $K_n$ be a covering of $\X$ such that
$\lamX(K_n) < \infty$ for all $n$.
Define $f_n(z) \defeq \lamXZ(K_n \given z) \in [0, \infty]$.
Then
\begin{equation}
\int_{\Z} f_n(z) \lamZ(\dz) = \lamX(K_n) < \infty,
\end{equation}
so $f_n < \infty$ for $\lamZ$-a.e.~$z$.
That is, $\lamZ(M_n) = 0$ for $M_n \defeq \set{z \in \Z \mid f_n(z) = \infty}$.
Put $M \defeq \cup_{n=1}^{\infty}M_n$, with $\lamZ(M) = 0$.
For any $z \notin M$ and every $n \ge 1$,
\begin{align}
\lamXZ(K_n \given z) = f_n(z) < \infty,
\end{align}
which shows that $\lamXZ(\cdot \given z)$ is $\sigma$-finite for every
$z \notin M$.

To establish \zcref[S]{thm:bayes-marginal-z}, for $A \in \SigZ$
\begin{align}
\muZ(A)
= \muX(t^{-1}(A))
&= \int_{t^{-1}(A)} \pX(x)\lamX(\d x) \\
&= \int_{\Z} \int_{t^{-1}(A)} \pX(x) \lamXZ(\dx \given z)\lamZ(\dz) \label{eq:achlorhydria} \\
&= \int_{A} \int_{t^{-1}(z)} \pX(x) \lamXZ(\dx \given z)\lamZ(\dz) \\
&\equiv \int_{A} \pZ(z) \lamZ(\dz),
\end{align}
where the second-to-last equality follows from the fact that $\lamXZ(\dx \given z)$
lives on the fiber $t^{-1}(z)$, so that the domain of integration in
\cref{eq:achlorhydria} is precisely
\begin{equation}
\set*{(z,x) \mid z \in \Z, x \in t^{-1}(A), x \in t^{-1}(z) } = \set*{(z,x) \mid z \in A, x \in t^{-1}(z)}.
\end{equation}
To show that $\pZ(z) \in (0, \infty)$, $\muZ$-almost surely, assume for
contradiction that
there exists a positive probability
event $N_\infty \subset \set{z \in \Z \mid \pZ(z) = \infty}$,
so that $\muZ(N_\infty) > 0$.
As $\muZ \ll \lamZ$, it follows that $\lamZ(N_\infty) > 0$,
and so
\begin{equation}
\muZ(N_\infty) = \int_{N_\infty} \pZ(z) \lamZ(\dz) = \infty,
\end{equation}
which cannot hold, as $\muZ$ is a probability measure.
That $N_0 \defeq \set{z \in \Z \mid \pZ(z) = 0}$ has probability
zero is immediate:
\begin{equation}
\muZ(N_0) = \int_{N_0} \pZ(z) \lamZ(\dz) = 0.
\end{equation}

To establish \cref{thm:bayes-posterior-x}, for $A \in \SigX$
\begin{align}
\muX(A)
= \int_{A} \pX(x)\lamX(\d x)
&= \int_{\Z} \int_{A} \pX(x)\lamXZ(\dx\given z) \lamZ(\dz) \\
&= \int_{\Z} \left[ \int_{A} \frac{\pX(x)}{\pZ(z)} \lamXZ(\dx\given z) \right] \pZ(z) \lamZ(\dz) \\
&= \int_{\Z} \int_{A} \frac{\pX(x)}{\pZ(z)} \lamXZ(\dx\given z) \muZ(\dz) \\
&\equiv \int_{\Z} \muXZ(A\given z) \muZ(\dz),
\end{align}
where the division by $\pZ(z)$ is secure, as we have shown that
$\pZ(z) \in (0,\infty)$, $\muZ$-almost surely.
\end{proof}

\getkeytheorem{thm:sf-local-base}

\begin{proof}
By \cref{lemma:c1-borel-rectifiable}, the set $\Z[\pi,a]$ is countably $\mpath$-rectifiable for
each $a \in \Apath$.
Because $\Apath$ is countable, $\Z[\pi] \defeq \cup_{a \in \Apath} \Z[\pi,a]$
is also countably $\mpath$-rectifiable.
The result follows from \cref{lemma:haus-restrict-rectifiable}.
\end{proof}

\getkeytheorem{thm:disintegrate-path}

\begin{proof}
Apply \cref{lemma:coarea-c1} with
$\X = \X[\pi,a]$,
$U = \U[\pi,]a$,
$\bar{f} = \obs*[\pi,a]$,
$f = \obs[\pi,a]$,
$\Z = \Z[\pi,a]$, and
$B = \B[\pi,a].$
\end{proof}

\getkeytheorem{thm:disintegration-local}

\begin{proof}
Using \cref{thm:disintegrate-path}, we have
\begin{equation}
\lamX[\pi](B)
=
\sum_{a\in \Apath} \Leb[\dpath](\B[\pi,a])
=
\sum_{a\in \Apath}
  \int_{\Z[\pi,a]}
    \int_{\obsI[\pi,a](z)}
      \frac{1_{B_{\pi,a}}(w)}{\Jac(w)}
        \Haus{\dpath-\mpath}[\dpath](\d w)
    \Haus{\mpath}[\npath](\d z).
\end{equation}
For each $a\in \Apath$, the coarea formula integrates over
$\Z[\pi,a] = \obs*[\pi,a](\U[\pi,a])$.
Since $\Z[\pi,a]\subseteq \Z[\pi]\defeq\bigcup_{a\in \Apath}\Z[\pi,a]$,
and the inner integral is zero for $z\notin \Z[\pi,a]$,
we may extend each outer integral to $\Z[\pi]$.
Thus, by Tonelli,
\begin{equation}
\sum_{a\in \Apath}\int_{\Z[\pi,a]} I_a(z)\Haus{\mpath}[\npath](\d z)
=
\int_{\Z[\pi]}\sum_{a\in \Apath} I_a(z)\Haus{\mpath}[\npath](\d z)
=
\int_{\Z[\pi]}\lamXZ[\pi](B \given z)\lamZ[\pi](\d z),
\end{equation}
where $I_a(z)$ is the inner fiber integral.
The last equality holds by the definitions of
$\lamXZ[\pi]$ and $\lamZ[\pi]$.
\end{proof}

\getkeytheorem{thm:lamZ-sigma-finite}

\begin{proof}
Observe the base measure is equal to a sum of restricted Hausdorff measures:
\begin{align}
\lamZ(B) = \sum_{n=0}^{\infty} \sum_{m=0}^{n} \Haus{m}[n][\Zeq{n}{m}]\!\left(B \cap \R[n]\right)
&& (B \in \SigZ).
\end{align}
We show that each summand is sigma-finite.
By \cref{lemma:c1-borel-rectifiable}, each $\Z[\pi,a]$ is countably $\mpath$-rectifiable,
which implies the same for $\Z[\pi]$.
Therefore $\Zle{n}{m}$ and $\Zeq{n}{m}$ are both countably $m$-rectifiable.
By \cref{lemma:haus-restrict-rectifiable}, the restricted measures
$\Haus{m}[n][\Zeq{n}{m}]$ are each sigma-finite.

We now construct a countable cover of $\Z$ by sets of finite $\lamZ$-measure.
The collection
\begin{equation}
\set*{\Zeq{n}{m} \mid n \ge 0, m = 0, \dots, n}
\end{equation}
is countable and pairwise disjoint.
For each $(n,m)$ fix a countable cover $\Zeq{n}{m} = \cup_{j=1}^{\infty} E_{n,m,j}$
such that $\Haus{m}[n](E_{n,m,j}) < \infty$.
As the strata are disjoint, we have $\lamZ(E_{n,m,j}) = \Haus{m}[n](E_{n,m,j}) < \infty$.
Together with the complement of
$\cup_{n,m} \Zeq{n}{m}$---which has $\lamZ$-measure zero---the
sets $(E_{n,m,j}\, ; n \ge 0, m = 0,\dots,n, j \ge 1)$ form a countable cover
of $\Z = \cup_{n=0}^{\infty}\R[n]$.
The conclusion follows.
\end{proof}

\getkeytheorem{prop:rd-local-global}

\begin{proof}
Fix a path $\pi \in \Pi$ and event $C \in \SigZ$.
As $\Z[\pi] \subset \R[\npath]$, it is sufficient to consider
the strata inside the $\npath$-summand.
Consider the decomposition
\begin{align}
C \cap \Z[\pi] = \cup_{m=0}^{\npath} C \cap \Z[\pi] \cap \Zeq{\npath}{m}.
\label{eq:rhyptic}
\end{align}
Consider the following two cases on $m$.
\begin{enumerate}[label=(Case \arabic*),wide]
\item For $\mpath < m$, as the $\Zeq{n}{m}$ are disjoint for each $n,m$
we have
\begin{align}
\Z[\pi] \subset \Zle{\npath}{\mpath} \subset \Zle{\npath}{m-1}
\implies
\Z[\pi] \cap \Zeq{\npath}{m} = \varnothing.
\end{align}

\item For $m < \mpath$, the set $\Zeq{\npath}{m}$ is countably $m$-rectifiable
by \cref{lemma:c1-borel-rectifiable}.
Therefore
\begin{equation}
\Haus{\mpath}[\npath]\!\left(C \cap \Z[\pi] \cap \Zeq{\npath}{m}\right) = 0,
\end{equation}
because $\Zeq{\npath}{m}$ is countably $m$-rectifiable,
which by \cref{lemma:rectifiable-haus-zero} has zero $\Haus{\mpath}$ measure.
\end{enumerate}

It follows that the only member in the decomposition of $C \cap \Z[\pi]$ in
\cref{eq:rhyptic} that has nonzero measure is the minimal stratum
$\Zeq{\npath}{\mpath}$.
Therefore
\begin{align}
\lamZ[\pi](C)
= \Haus{\mpath}[\npath](C \cap \Z[\pi])
= \Haus{\mpath}[\npath]\!\left(C \cap \Z[\pi] \cap \Zeq{\npath}{\mpath}\right).
\end{align}
But on $\Zeq{\npath}{\mpath}$, the global base measure $\lamZ$ is precisely
$\restr{\Haus{\mpath}[\npath]}{\Z{\npath}{\mpath}}$.
Therefore
\begin{align}
\lamZ[\pi](C)
= \lamZ\left(C \cap \Z[\pi] \cap \Zeq{\npath}{\mpath}\right)
= \int_{C} \ind[\Z[\pi]]{z} \ind{\rank(z) = \mpath} \lamZ(\d z).
\end{align}

\end{proof}

\getkeytheorem{thm:disintegration-global}

\begin{proof}
The result follows from \cref{thm:disintegration-local,prop:rd-local-global} and Tonelli's theorem:
\begin{align}
\lamX(B)
&\defeq \sum_{\pi \in \Pi} \lamX[\pi](\B[\pi]) \\
&= \sum_{\pi \in \Pi} \int_{\Z[\pi]} \lamXZ[\pi](\B[\pi] \mid z) \lamZ[\pi](\d z) && \mbox{(\cref{thm:disintegration-local})} \\
&= \sum_{\pi \in \Pi} \int_{\Z} \rd(z) \lamXZ[\pi](\B[\pi] \mid z) \lamZ(\d z) && \mbox{(\cref{prop:rd-local-global})} \\
&= \int_{\Z} \sum_{\pi \in \Pi} \rd(z) \lamXZ[\pi](\B[\pi] \mid z) \lamZ(\d z) && \mbox{(Tonelli's Theorem)}.
\end{align}
\end{proof}

\getkeytheorem{thm:importance-sampling}

\begin{proof}[Proof of \Cref{thm:importance-sampling}]
To fix notation, we are importance sampling the posterior distribution
$\muXZ(\cdot \given z)$ using a proposal distribution $\QXZ$ that performs ancestral
sampling.
For any pair $(\pi,x)$ of a path and addressed trace, let $l'(\pi,x)$ denote the product
of the primitive densities of the constrained variables along the path.
Let $F_z \defeq \obsI(z) = \set{(\pi,x) \mid \obs[\pi](x) = z}$ denote the fiber supporting the posterior
distribution, and let $\QXZ' \defeq \restr{\QXZ}{F_z}$ be the restriction of the proposal
to this fiber, so that $\QXZ' \ll \lamXZ*(\cdot \given z)$.
By the ancestral sampling property, we have
\begin{align}
\pX(\pi,x) = l'(\pi,x) \frac{\d \QXZ'}{\d\lamXZ*(\cdot\given z)}(\pi,x)
&& ((\pi,x) \in F_z).
\end{align}
Therefore, $\muXZ(\cdot\given z)\ll \QXZ'$ and
the importance weight
\begin{align}
w(\pi,x)
=
\frac{\d \muXZ(\cdot \given z) }{ \d \QXZ'  }(\pi,x)
=
  {\frac{\d \muXZ(\cdot \given z) }{ \d \lamXZ*(\cdot\given z)  }(\pi,x)}
  \Bigg/
  {\frac{\d \QXZ'(\cdot \given z) }{ \d \lamXZ*(\cdot\given z)  }(\pi,x)}
= \frac{l'(\pi,x) \rd(z)}{\pZ(z)}
\end{align}
is proportional to the usual likelihood weight from ancestral sampling, multiplied
by the rank-selection indicator.
Because $\muXZ(\cdot \given z)$ is supported on $F_z$, we can extend this
density by $0$ off $F_z$ to obtain the density with respect to the full proposal:
\begin{align}
\frac{\d \muXZ(\cdot \given z) }{ \d \QXZ(\cdot)  }(\pi,x)
= \frac{\d \muXZ(\cdot \given z) }{ \d \QXZ'(\cdot)  }(\pi,x)
  \frac{\d \QXZ'(\cdot) }{ \d \QXZ(\cdot)  }(\pi,x)
= \frac{l'(\pi,x) \rd(z)}{\pZ(z)} \ind{\obs[\pi](x)=z}.
\end{align}
Let $l(\pi,x) \defeq l'(\pi,x)\ind{\obs[\pi](x)=z}$.
We write the unnormalized weight as
\begin{align}
\widetilde{w}(\pi,x) \defeq l(\pi,x) \rd(z) = l(\pi,x)1(\mpath[\pi]=\rank(z)),
\end{align}
noting that as $l(\pi,x)>0$ only if $\obs[\pi](x) = z$, the factor
$\ind{z\in \Z[\pi]}$ in $\rd(z)$ may be omitted.

We seek to estimate
\begin{align}
\Expect[\muXZ(\cdot \given z)]{f(\bm\pi,X)}
&= \int_{\X} f(\pi,x)\,\muXZ(\d(\pi,x)\mid z).
\end{align}

The standard importance-sampling identity gives
\begin{align}
\Expect[\QXZ]{\widetilde{w}(\bm\pi,X)}
= \Expect[\QXZ]{l(\bm\pi,X)\ind{\mpath[\bm\pi]=\rank(z)}}
= \pZ(z) > 0.
\end{align}
Thus,
the event
\begin{align}
A \defeq \set{\widetilde{w}(\bm\pi,X)>0} \equiv \set{l(\bm\pi,X)>0, \mpath[\bm\pi]=\rank(z)}
\end{align}
has positive $\QXZ$-probability.

For the i.i.d.~sequence
$((\bm\pi_j, X_j))_{j=1}^{\infty} \sim \QXZ[\infty]$
of importance samples, define the events
\begin{equation}
A_j\defeq \set{l(\bm\pi_j,X_j)>0,\ \mpath[\bm\pi_j]=\rank(z)}
\end{equation}
The events $A_j$ are independent and identically distributed, so
\begin{equation}
\sum_{j=1}^\infty \QXZ[\infty](A_j) = \sum_{j=1}^\infty \QXZ(A) = \infty.
\end{equation}
By the second Borel--Cantelli lemma, $A_j$ occurs i.o., $\QXZ[\infty]$-a.s.
In particular, the hitting time is finite:
\begin{align}
\tau\defeq \inf\set{j\ge 1\mid l(\bm\pi_j,X_j)>0, \mpath[\bm\pi_j]=\rank(z)} < \infty,
&& \QXZ[\infty]\mbox{-a.s.}
\end{align}
Moreover, if $l(\bm\pi_j,X_j)>0$, then $z\in \Z[\pi_j]$, so by definition
of $\rank(z)$, we must have $\mpath[\pi_j] \ge \rank(z)$.
Therefore, $\rank_N(z)=\rank(z)$ for every $N\ge \tau$,
so for all sufficiently large $N$, we have
\begin{align}
\widehat{\mathbb E}^z_N[f]
= \frac
  { \sum_{j=1}^N f(\bm\pi_j,X_j)l(\bm\pi_j,X_j)\ind*{\mpath[\bm\pi_j]=\rank(z)} }
  { \sum_{j=1}^N l(\bm\pi_j,X_j)\ind*{\mpath[\bm\pi_j]=\rank(z)} },
&& \QXZ[\infty]\mbox{-a.s.}
\end{align}

As $f$ is bounded and
$\Expect[\QXZ]{\widetilde{w}(\bm\pi,X)}=\pZ(z)<\infty$, the strong law of large numbers
gives
\begin{align}
\frac{1}{N} \sum_{j=1}^N f(\bm\pi_j,X_j)l(\bm\pi_j,X_j)\ind*{\mpath[\bm\pi_j]=\rank(z)}
&\longrightarrow \Expect[\QXZ]{f(\bm\pi,X)l(\bm\pi,X)\ind*{\mpath[\bm\pi]=\rank(z)}},
\\
\frac{1}{N} \sum_{j=1}^N l(\bm\pi_j,X_j)\ind*{\mpath[\bm\pi_j]=\rank(z)}
&\longrightarrow \Expect[\QXZ]{l(\bm\pi,X)\ind{\mpath[\bm\pi]=\rank(z)}} = \pZ(z).
\end{align}
in the limit as $N \to \infty$, $\QXZ[\infty]\mbox{-a.s.}$

Because the limit in the denominator is positive, the ratio converges:
\begin{align}
\lim_{N \to \infty} \widehat{\mathbb E}^z_N[f]
=
\frac
  { \Expect[\QXZ]{f(\bm\pi,X)l(\bm\pi,X)\ind{\mpath[\bm\pi]=\rank(z)}} }
  { \pZ(z) },
&&
\QXZ[\infty]\mbox{-a.s.}
\end{align}

Finally, using the importance-sampling identity once more,
the right-hand side is
\begin{align}
\frac{ \Expect[\QXZ]{f(\bm\pi,X) l(\bm\pi,X)\ind{\mpath[\bm\pi]=\rank(z)}} }{ \pZ(z) }
&=
\frac{\pZ(z) \thinspace \Expect[\muXZ(\cdot \given z)]{f(\bm\pi,X)}} { \pZ(z) }
\\
&=\int_{\X} f(\pi,x)\,\muXZ(\d(\pi,x)\mid z).
\end{align}
\end{proof}


\section{Singular Primitive Distributions}
\label{appx:primitives}

In our formalism from \cref{sec:lang-primitives}, the primitive
distributions are by assumption either discrete or absolutely continuous
with respect to Lebesgue measure.
A more flexible approach would be to allow primitive distributions that
are supported on parametrized lower dimension manifolds, for example the
boundary of a circle with radius $\theta$ in $\R[2]$.
Our semantics avoids this additional generality because singular primitive
distributions introduce significant complexities.
\begin{itemize}

\item A singular primitive distribution may not be dominated
by any fixed $\sigma$-finite  base measure.

\item A singular primitive distribution must carry tangent-space data for
its support in order to compute the appropriate Jacobian corrections.
\end{itemize}
Resolving these challenges requires adding significant geometric structure
to the semantics that is not canonically part of the program text.
We elaborate on these two issues.

\paragraph{Base Measure}

A primitive distribution on a lower-dimensional manifold
may not admit a unified base measure on which to define its densities.
For example, suppose $D$ is a primitive such that $\muD(\cdot \given \theta)$
is the uniform distribution on the circle $S_\theta=\set{x \in \R[2] : \abs{x} = \theta}$.
Suppose $\lamD$ is a measure on $\R[2]$ such that
$\muD(\cdot \given \theta) \ll \lamD$ for every $\theta > 0$.
Then $\lamD(S_\theta)>0$ for uncountably many pairwise disjoint circles
$S_\theta$, so $\lamD$ is not sigma-finite \citep[Proposition 215B]{fremlin2010}.

As there is no sigma-finite global base measure $\lamD$, we may instead
consider a parameterized family $\lamD(\cdot \given \theta) \equiv \Haus{1}[2][S_\theta]$
of sigma-finite measures for $\theta \in \Theta$.
In this case, the resulting parametrized densities $\pD(x \given \theta)$
no longer serve as meaningful likelihoods, because they are not defined
with respect to a fixed base measure $\lamD$ that is independent of $\theta$.
If the base measure depends on $\theta$, then
the densities are not comparable across different parameter values.
For example, the rescaling
\begin{equation}
\lamD(\cdot \given \theta) \rightsquigarrow c(\theta)\lamD(\cdot \given \theta), \qquad
\pD(\cdot\given \theta) \rightsquigarrow \pD(\cdot\given\theta)/c(\theta),
\end{equation}
for some positive function $c$ leaves every distribution $\mu(\cdot \given \theta)$
unchanged, but changes all the likelihood calculations involving $\theta$.
Therefore, handling primitives on lower-dimensional manifolds must
either fix the support manifold and let parameters change only the density
on that support, or provide a different representation such as sampling in
intrinsic coordinates followed by a deterministic embedding (which may as well
be performed in the probabilistic program itself).

In our semantics, the assumption that $D$ is either continuous or discrete
(with counting or Lebesgue base measure) avoids these complexities,
because the counting measure on every countable set is sigma-finite
and $(\R[n], \Borel(\R[n]), \Leb[n])$ is sigma-finite.
Even if the support of $\muD(\cdot \given \theta)$ varies with $\theta$,
all the distributions remain dominated by $\lamD$.
For example, if $D$ is uniform over $[0,\theta]$, then
$\pD(x \given \theta) = 1 / \theta \cdot \ind[[0,\theta]]{x}$,
and this density is valid with respect to the fixed base measure $\lamD = \Leb$.

\paragraph{Coarea Jacobian}

In a language with primitives over lower-dimensional manifolds,
the domain of the observation maps
$\obs[\pi,a]:\X[\pi,a]\to\R[\npath]$ for each $\pi \in \Pi$ and $a \in \Apath[\pi]$
is not guaranteed to be a positive measure set.
This generalization is compatible in principle with
\citet[Theorem 3.2.22]{federer1996}, which only requires the domain of
the function to be countably rectifiable.

In this setting, however, the coarea Jacobian $\Jac[\pi](x)$ in \Cref{eq:coarea-jacobian}
is not a straightforward function of the standard derivative matrix
$\obs*[\pi]'(x)$ in \Cref{eq:deriv-matrix}; it must measure volume change along $\X[\pi,a]$.
\begin{itemize}
\item If $\X[\pi,a]$ is a $d$-dimensional embedded manifold, then the
ordinary derivative must be restricted to the tangent space
$T_x\X[\pi,a]$.
Equivalently, if $Q_x$ is an orthonormal matrix whose columns form a basis
for $T_x\X[\pi,a]$, then the relevant coarea Jacobian is the tangential
Jacobian
\begin{equation}
  \Jac[\pi](x)
  =
  \prod_{i=1}^{\mpath}
  \sigma_i\!\left(\obs*[\pi,a]'(x)Q_x\right),
\end{equation}
where $\mpath$ is the rank of the restricted derivative
$\obs*[\pi,a]'(x)|_{T_x\X[\pi,a]}$.
With a non-orthonormal tangent basis, this same correction can be written
using the Gram-determinant ratio, similarly to the manifold
change-of-variables formula of \citet[Theorem~1, Eq.~(3)]{radul2021base}.

\item If $\X[\pi,a]$ is not a manifold but only countably $d$-rectifiable,
then the tangent space must be replaced by the approximate tangent space
and the derivative by the approximate differential
\citep[3.2.16--3.2.19]{federer1996}.
The corresponding Jacobian is that of the approximate differential, which
appears in the coarea formulas for rectifiable sets~\citep[3.2.20, 3.2.22]{federer1996}.
\end{itemize}

Our semantics avoids these complexities.
By assuming that $\X[\pi,a]$ has nonzero Lebesgue measure, we ensure that
the standard derivative matrix $\obs*[\pi,a]'(w)$ contains all the
information needed to compute the coarea
Jacobian~\cref{eq:coarea-jacobian}.
The underlying partial derivatives in \cref{eq:deriv-matrix} can be obtained by
applying AD to the $\textbf{return}$ expression.
In the general setting, however, computing the coarea Jacobian requires not
only added structure about the (approximate) tangent spaces, but also a
``tangent-aware'' AD method that propagates tangent
bases through deterministic program expressions.

}

\end{document}